\documentclass[letterpaper,twocolumn,10pt]{article}
\usepackage{usenix2019,epsfig,endnotes}

\usepackage{booktabs}
\usepackage{enumitem}
\usepackage{subcaption}
\usepackage{tabularx}
\usepackage{graphicx} 
\usepackage{color}
\usepackage{hyperref}
\usepackage{amsmath}
\usepackage{amsthm}

\newcommand{\compactcaption}[1]{\vspace{-0.90em}\caption{#1}\vspace{-1.05em}}

\newtheorem{theorem}{Theorem}

\begin{document}

\date{}

\title{\Large \bf Expander Datacenters: From Theory to Practice}

\author{
 {Vipul Harsh, Sangeetha Abdu Jyothi, Inderdeep Singh, P. Brighten Godfrey} \\
   UIUC \\
}


\maketitle

\thispagestyle{empty}





\subsection*{Abstract}
Recent work has shown that expander-based data center topologies are robust and can yield superior performance over Clos topologies. However, to achieve these benefits, previous proposals use routing and transport schemes that impede quick industry adoption. 
In this paper, we examine if expanders can be effective for the technology and environments practical in today's data centers, including use of traditional protocols, at both small and large scale, while complying with common practices such as over-subscription.
We study bandwidth, latency and burst tolerance of topologies, highlighting pitfalls of previous topology comparisons. We consider several other metrics of interest: packet loss during failures, queue occupancy and topology degradation. 
Our experiments show that expanders can realize $3\times$ more throughput 
than an equivalent fat tree, and $1.5\times$ more throughput than an equivalent leaf spine topology, for a wide range of scenarios, with only traditional protocols. 
 We observe that expanders achieve lower flow completion times, are more resilient to bursty load conditions like incast and outcast and degrade more gracefully with increasing load. 
Our results are based on extensive simulations and experiments on a hardware testbed with realistic topologies and real traffic patterns.


\section{Introduction}

Modern data centers (DCs) need to be designed for high performance and for high burst tolerance. To achieve the same, hyperscale DC architectures are based on Clos networks or Fat-trees~\cite{AlFares08} that provide full bandwidth (modulo oversubscription) e.g.  Google's Jupiter ~\cite{JupiterRising15}, Facebook's DC fabric~\cite{fbDcFabric} and VL2~\cite{vl209}. Small and medium DCs, where the 3-tier Clos topology is overkill, are commonly based on two-tiered leaf-spine designs. 
Recently, alternative topologies based on expander graphs\footnote{Expanders are an extensively studied class of graphs that are in a sense maximally or near-maximally well-connected.}, such as Jellyfish~\cite{singla12}, Xpander~\cite{valadarsky16}, Slimfly~\cite{besta14} and LongHop~\cite{tomic13}, have been proposed for high performance. 

To extract maximum benefits from expander networks, previous proposals have relied on non-traditional protocols. Multipath TCP (MPTCP)~\cite{MPTCP} and k-shortest path routing~\cite{kShortestPaths} was used in Jellyfish~\cite{singla12} and Xpander~\cite{valadarsky16}. Kassing et. al.~\cite{kassing17} demonstrated benefits of expanders using a hybrid of ECMP and Valiant load balancing when combined with flowlet switching~\cite{flowletSwitching1,flowletSwitching2}. While all the above schemes are effective, they are not very widely used in today's data centers and require additional mechanisms. Specifically, K-shortest path routing is feasible with today's hardware but requires both control plane and data plane modifications; MPTCP requires operating system support and configuration which is often out of the control of the data center operator; and flowlet switching depends on flow size detection and dynamic switching of paths. We argue that dependence on these protocols impedes quick industry adoption.




In this paper, we examine if expander networks retain their performance advantages in environments practical for today's data centers. Towards this end, we consider several factors. Traffic engineering should ideally not require host transport modifications. Routing should ideally be based on simple oblivious\footnote{Oblivious routing refers to path selection independent of the current traffic pattern.} schemes using commonly available mechanisms like shortest path routing, ECMP, segment routing, and MPLS. 
And finally, whereas (to the best of our knowledge) all past work on expander DCs has focused on comparing with larger three-tiered Clos networks, it is important to understand whether these topology designs can be applied to the much more numerous leaf-spine networks of small- and medium-scale DCs, with realistic oversubscription ratios.

To analyze performance for a wide range of cases, we introduce the $C-S$ model as a synthetic workload (Section~\ref{cs model}). The $C-S$ model captures several frequently occurring application patterns in distributed systems besides a range of skewed and uniform patterns. We consider throughput and flow completion times as the main performance metrics.  We test burst tolerance for different types of incast and outcast patterns (modeled as specific points in the $C-S$ model) and present experiments with a publicly available Facebook DC traffic trace~\cite{fbTrafficPattern}. We also present experiments on a physical testbed by emulating modestly sized topologies (up to 48 servers). We study several other metrics of interest- queue occupancy, traffic loss due to failures, topology degradation with increasing load and fraction of long distance links. We base our experiments on shortest path routing and TCP in the transport layer which constitute the most common setting in DCs. Additionally, we compare k-disjoint path routing, which can be implemented using segment routing~\cite{segment_routing} (Section~\ref{kshortimpl}).  We use realistic baseline topologies, that are inspired by industry recommended sizing configurations~\cite{AristaBestPractices}.

Our results show that expanders can offer a very large performance improvement over Clos topologies, using only currently-supported protocols, for a wide range of traffic patterns in the $C-S$ model, particularly with realistic over-subscription ratios (e.g., $2\times$-$4\times$). Moreover, this performance advantage is present even for relatively small leaf-spine networks, such as an 8-leaf, 2-spine hardware testbed.

We found these results surprising, given that past work has required more advanced protocols, and has found generally larger benefits with larger scale~\cite{jyothi2016}. There are several factors at play influencing these performance differences. Past work has discussed the expander's advantage in lower mean path length (hence, delivering each packet requires less total bandwidth), and the disadvantage that its more complex structure means it will not attain its full potential without special routing protocols. However, we found two other effects to be critical. First, by moving from a two-tiered design to a flat design (while using the same physical equipment), expanders increase the bandwidth exiting each rack. In fact, we show analytically that the ToR over-subscription is $2\times$ less in expanders than leaf-spines, irrespective of the number of leaf and spine switches, and is $4\times$ less in a 3-tier fat tree (Section~\ref{leaf spine comparisons}). This does not mean the expander always gets $2\times$ higher throughput. But when two factors are combined -- (1) over-subscription causing a bottleneck exiting the leaf-spine's racks, and (2) a skewed traffic pattern causing this bottleneck at a minority of racks -- the expander is effectively able to mask the over-subscription and behave closer to a non-blocking network. The oversubscribed setting also happens to be the more realistic scenario and has been overlooked in past work. 

Second, use of the expander's capacity is more efficient in oversubscribed networks. If the goal is to obtain a non-blocking (or \emph{full-bandwidth}) network as in many recent studies~\cite{singla12,kassing17}, expanders reach that level of performance only after hitting a regime of diminishing returns as more switches are added (Figure~\ref{fig:oversubscriptionCrucial}) since the DC is mostly host-bottlenecked by that point. Even with that inefficiency, expanders may edge out Clos networks in cost, but we find the advantage is larger in the oversubscribed regime.

A summary of our results are as follows:
\begin{itemize}[leftmargin=*]
\itemsep0em
\item We show that expanders can yield as much as $3$-$4\times$ higher throughput than fat trees and up to $2\times$ that of leaf-spine networks for a wide range of traffic patterns, under realistic over-subscription ratios and with practical protocols.
\item We demonstrate that the improvement in performance does not come at the cost of fairness; fairness in the expander was at least as good as in fat trees (Figure~\ref{fig:fairness}).
\item We show that for a publicly available DC traffic trace~\cite{fbTrafficPattern}, expanders result in significantly smaller flow completion times at the tail and degrade more gracefully than fat trees as the network is increasingly stressed. 
\item We show that expanders result in comparable amounts of transient packet loss to fat trees during link and switch failures. 
\item We demonstrate that expanders are more resilient to bursty load situations like incasts and outcasts resulting in smaller flow completion times and smaller queue sizes. 
\item We present modestly-sized experiments for topologies up to 48 servers emulated on a hardware testbed. Our hardware experiments compare throughput, flow completion times and the job completion time of a heavy shuffle operation and confirm our simulation results.
\end{itemize}

Our simulations are based on the htsim~\cite{htsim} and ns3~\cite{ns3} simulators. We have made our code publicly available and provided scripts to reproduce experimental results for any topology (obfuscated for review). Overall, we believe our results show that expanders are highly promising for deployment using today's hardware, even at moderate scale. Indeed, this scale may be especially practical in the near term, since concerns over wiring complexity are less relevant.

\begin{figure}
\centering
\includegraphics[height=1.5in, width=1.8in]{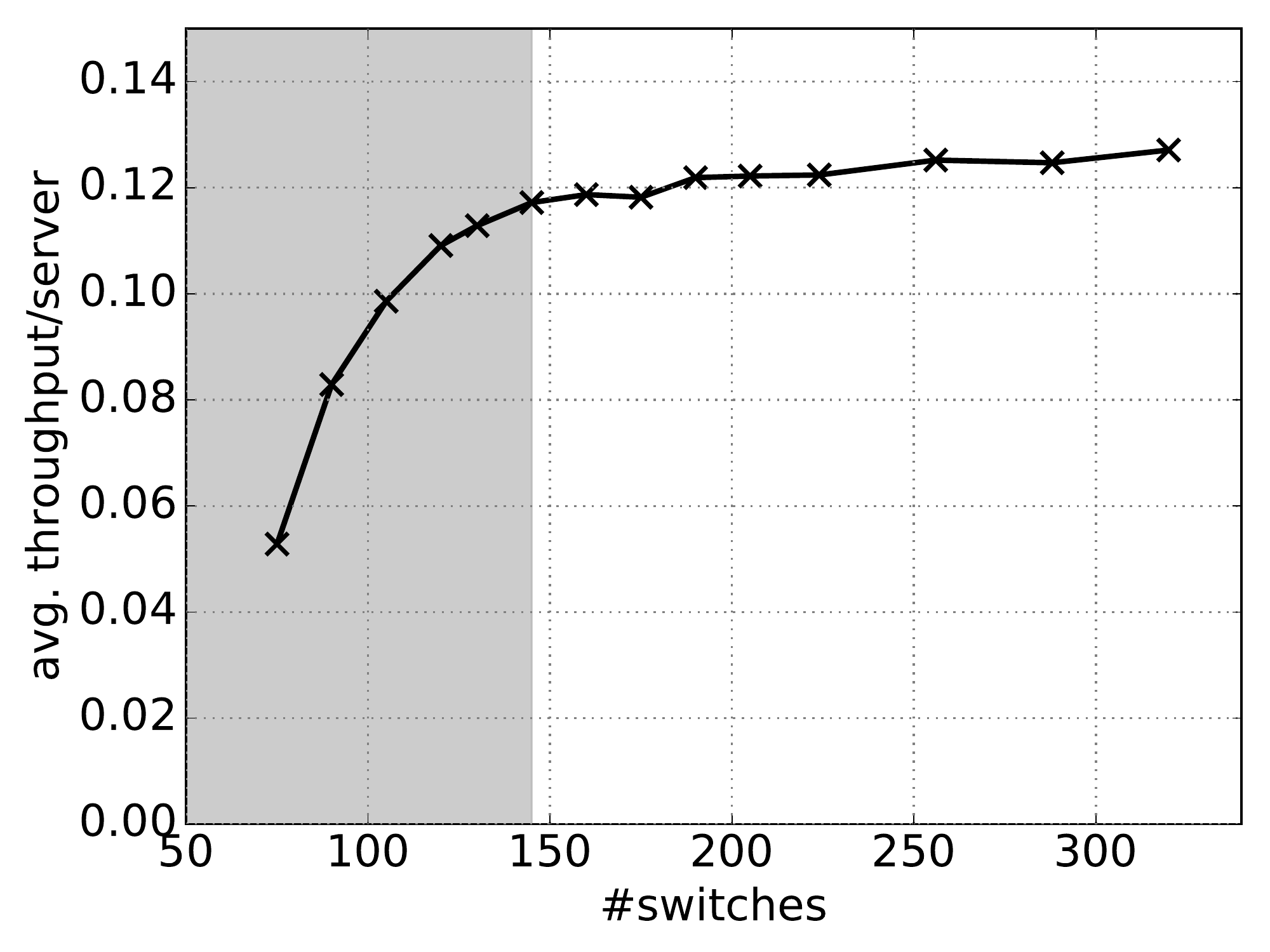}
\compactcaption{Throughput for skewed traffic for various sizes of random graphs supporting an equal number of servers. The limiting bottleneck is the network in the shaded region, and servers elsewhere, hence the flattening of the curve.}
\label{fig:oversubscriptionCrucial}
\end{figure}

\section{Background}
Expanders are a class of graphs that have high expansion ratios ~\cite{kowalski2013expander} (defined as the ratio of the outgoing edges from a set of nodes in a graph to the size of the set). Expander graphs provably~\cite{kowalski2013expander} have good connectivity and high bisection bandwidth, hence are a good option for datacenter networks. Singla et. al~\cite{singla12} first showed that random graphs (Jellyfish) are competitive to a Clos-based topology for datacenter networks. They demonstrated in ~\cite{singla14} that random graphs are flexible to accommodate for heterogeneous switch hardware and come close to optimal performance for uniform patterns under a fluid flow routing model.

Asaf et. al~\cite{valadarsky16} proposed the Xpander topology as a cabling-friendly alternative to Jellyfish, while matching its performance. Both Jellyfish and Xpander used K-shortest path routing and MPTCP~\cite{MPTCP} in the transport layer to demonstrate performance advantages. The Xpander topology has larger number of long distance links than the random graph, albeit with more structure allowing cables to be bundled together neatly. Nevertheless, the large number of long distance links hinders manageability. 
In Section~\ref{cross links}, we discuss a cabling solution to decrease the number of long distance links without compromising on expander quality. 

Kassing et. al~\cite{kassing17} demonstrated that expander networks can outperform fat trees for skewed traffic patterns using a combination of VLB, shortest path routing and flowlet switching~\cite{flowletSwitching1,flowletSwitching2}. More importantly, comparisons in Jellyfish~\cite{singla12,singla14}, Xpander~\cite{valadarsky16} and a recent evaluation~\cite{kassing17} suffer from bottleneck at the end hosts. Furthermore, reliance on non-traditional protocols restricts the practical realization of these proposals.  They also showed that under a fluid flow model where the host capacity is not modeled (and hence problems related to comparisons with full bandwidth do not arise), expanders can achieve around $3\times$ more throughput than fat trees. Their work left an open question of designing better routing algorithms to realize this with oblivious routing schemes.  In our work we show that in fact traditional routing algorithms suffice. However, the benefits are apparent only when the network is oversubscribed, rather than when the bottleneck lies primarily with hosts.


Jyothi et. al~\cite{jyothi2016} showed that under the fluid flow model with ideal routing, the random graph outperforms the fat tree for near-worst case traffic matrices. The worst case traffic matrix is a permutation matrix since any traffic matrix in the hose model~\cite{hoseModel} can be written as a convex combination of all permutation matrices (using Birkhoff-von Neumann's Theorem~\cite{birkhoffVonNeumann}). Although computing the exact worst case permutation matrix is computationally hard, longest matching traffic patterns are a good candidate. The comparisons in ~\cite{jyothi2016} are insightful since they don't model the server-to-switch bottleneck and hence avoid the diminishing returns of the full-bandwidth scenario. However, the use of optimal routing and an idealized fluid-flow model does not give insights about the feasibility of practical routing schemes.

These works leave substantial room for improving the practicality of using expanders for data centers. It's established that expanders can achieve superior performance; our goal is to examine if the gains still hold when used with traditional protocols. Answering that question is primarily a measurement study, guided towards practical settings. We next describe our methods for this study.





\section{Experimental Setup}

\subsection{Topologies}
\vskip -0.15 cm
\textbf{Fat trees (k)}:
We used standard fat trees~\cite{AlFares08} with $k$ ports, oversubscribed at ToR by 1,2 and 4$\times$. For oversubscribing, we assume that each ToR supports $4 \times$ more servers than a full bandwidth fat tree, if the oversubscription ratio is $4$. We used $k=16$ (320 switches, 4096 servers, $4\times$ oversubscribed) for simulations in htsim~\cite{htsim} and $k=10$ (125 switches, 1000 servers, $4\times$ oversubscribed) for simulations in ns3.

\textbf{Leaf-spine (x, y)}:
We consider the following leaf spine topology with switch degree $(x+y)$ for arbitrary parameters $x$ and $y$, 
\begin{itemize}
\itemsep0em
\item There are $y$ spine switches, each of them connect to all leaf switches.
\item There are $(x+y)$ leaf switches, each connect to all $y$ spine switches.
\item Each leaf switch is connected to $x$ servers.
\end{itemize}

As per recommended industry practices~\cite{AristaBestPractices}, we choose an oversubscription ratio $= x/y = 3$ with $x=24, y=8$ for experiments in the $C-S$ model. The leaf spine topology built this way with parameters $x=48$ and $y=16$ is exactly same as an actual industry-recommended configuration~\cite{AristaBestPractices}. We chose this recommended configuration in part because it uses leaf and spine switches with the same line speed, making comparisons more straightforward; we leave heterogeneous configurations to future work, but expect similar results.

\textbf{Random Graph}:
To compare a random graph topology with a leaf-spine, we build a random graph with the exact same equipment by rewiring the baseline leaf spine topology, redistributing servers equally across all switches (including switches that previously served as spines) and applying a random graph to the remaining ports.

To compare a random graph with a full bandwidth fat tree, we build a random graph with the exact same equipment the fat tree, redistributing servers across all switches (including switches that constitute the upper layer core and aggregation layer) and applying a random graph to the remaining ports.  To convert this to an oversubscribed topology, we add more servers to every rack of the random graph, distributing servers evenly across all racks. (Note that this follows the same over-subscription procedure as the fat tree except that it distributes the additional server ports across all switches rather than only the fat tree's ToRs.)

Note that other expanders have very similar performance characteristics to the random graph~\cite{valadarsky16} and hence our result applies to all high-end expanders. In our text, we use the shorthand RRG for Random Regular Graph. 


\begin{figure*}
  \hskip -0.1 cm
  \begin{minipage}[b]{0.32\textwidth}
    \centering
    \begin{subfigure}[b]{\linewidth}
    \includegraphics[width=0.95\textwidth]{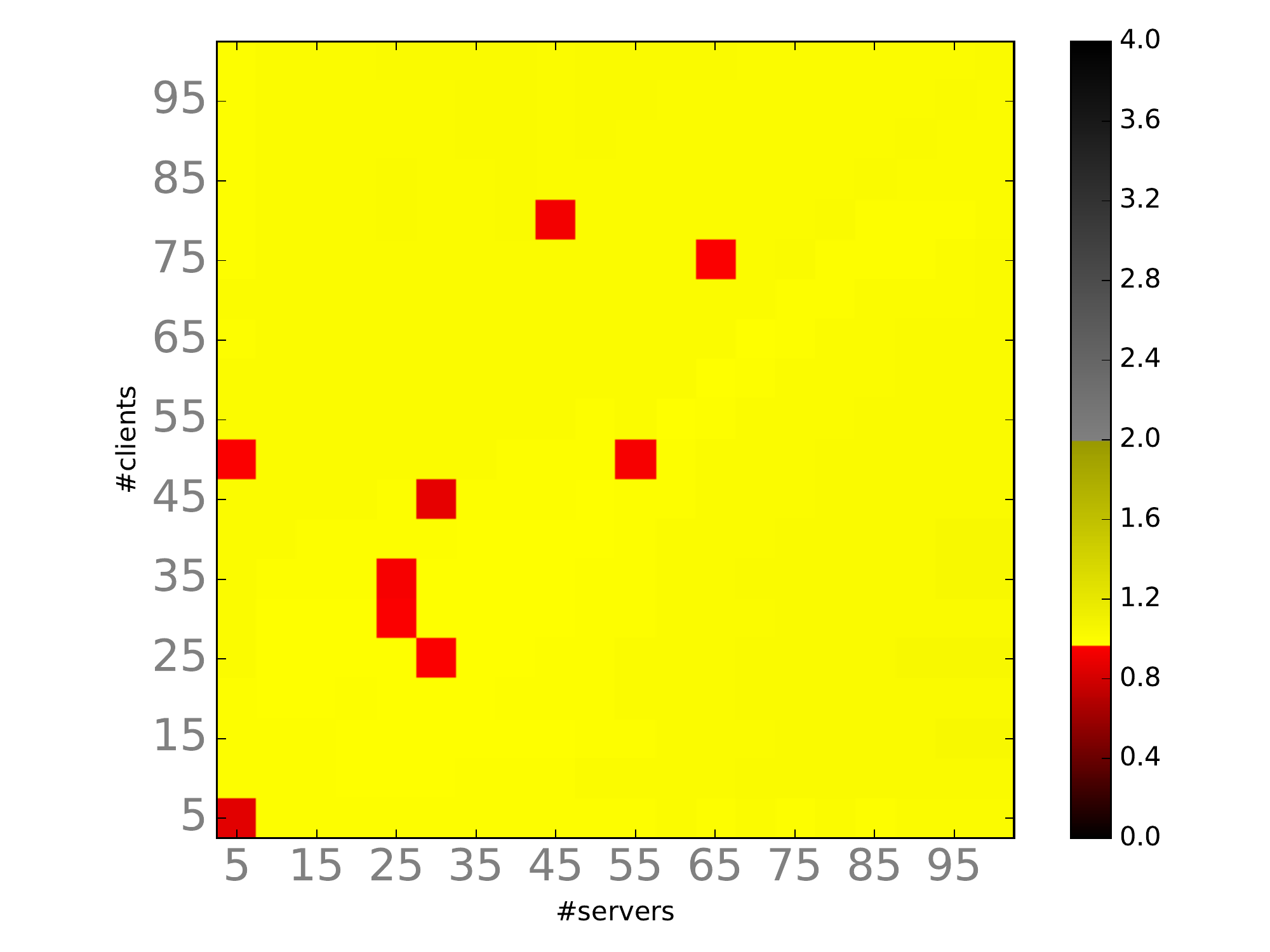}
    \caption{Small runs ECMP, OS=1}
    \label{fig:ecmp cs small os1}    
    \end{subfigure}
    \begin{subfigure}[b]{\linewidth}
    \includegraphics[width=0.95\textwidth]{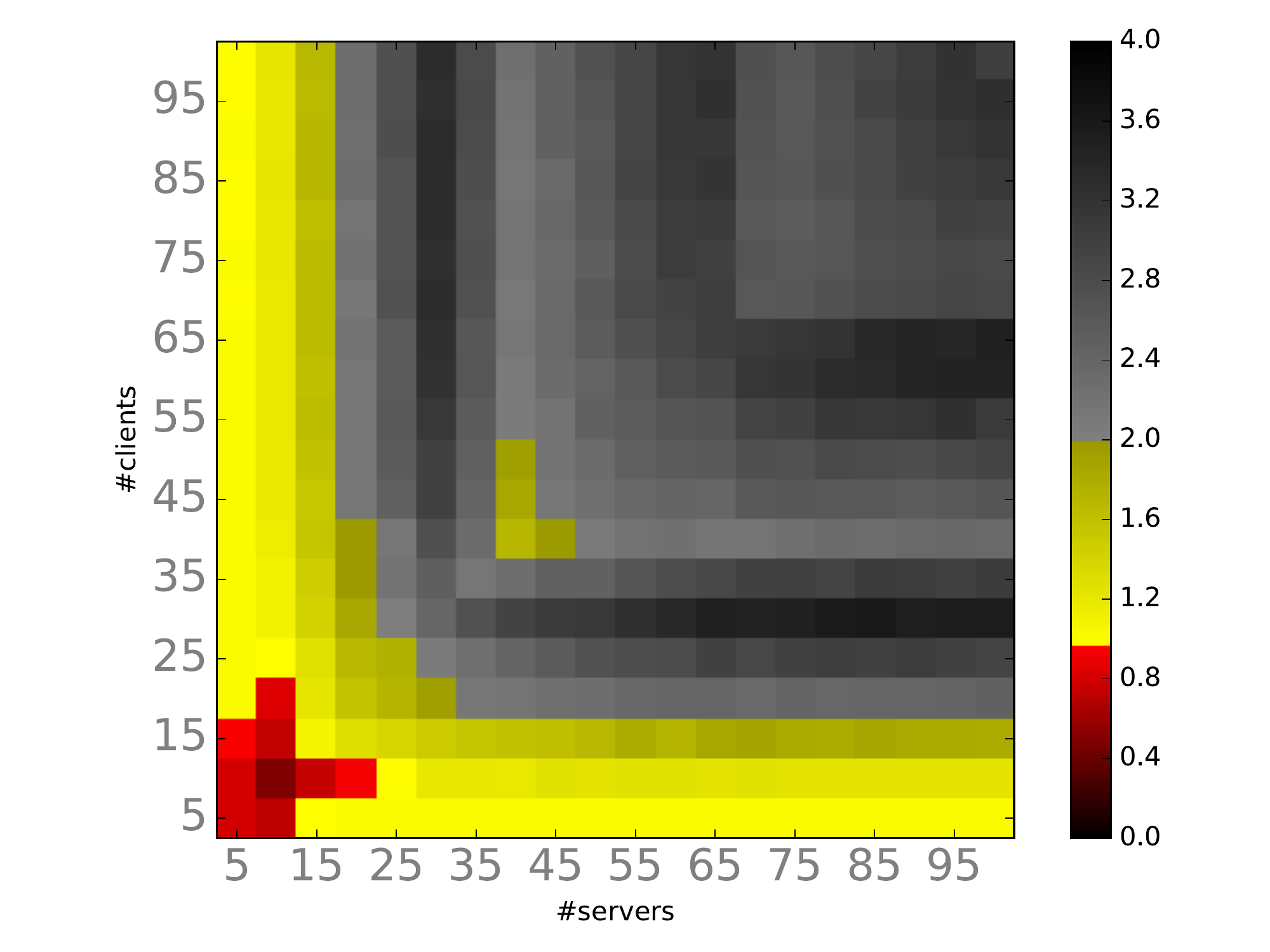}
    \caption{Small runs K Shortest paths, OS=4}
    \label{fig:kshort cs small}
    \end{subfigure}
  \end{minipage}
  \begin{minipage}[b]{0.32\textwidth}
    \centering
    \hskip -0.1 cm
    \begin{subfigure}[b]{\linewidth}
      \includegraphics[width=0.95\textwidth]{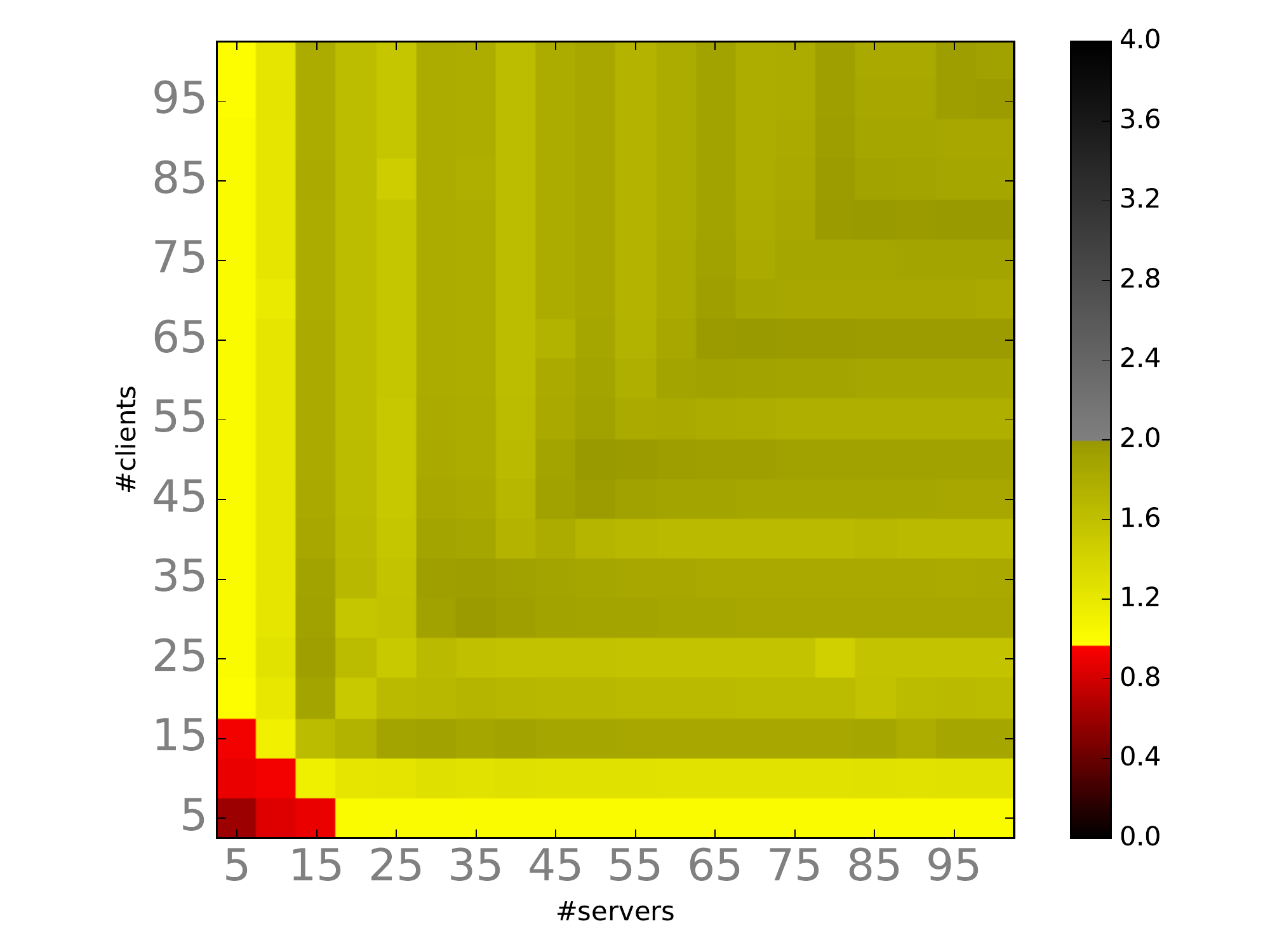}
      \caption{Small runs ECMP, OS=2}
      \label{fig:ecmp cs small os2}     
    \end{subfigure}
    \begin{subfigure}[b]{\linewidth}
      \includegraphics[width=0.95\textwidth]{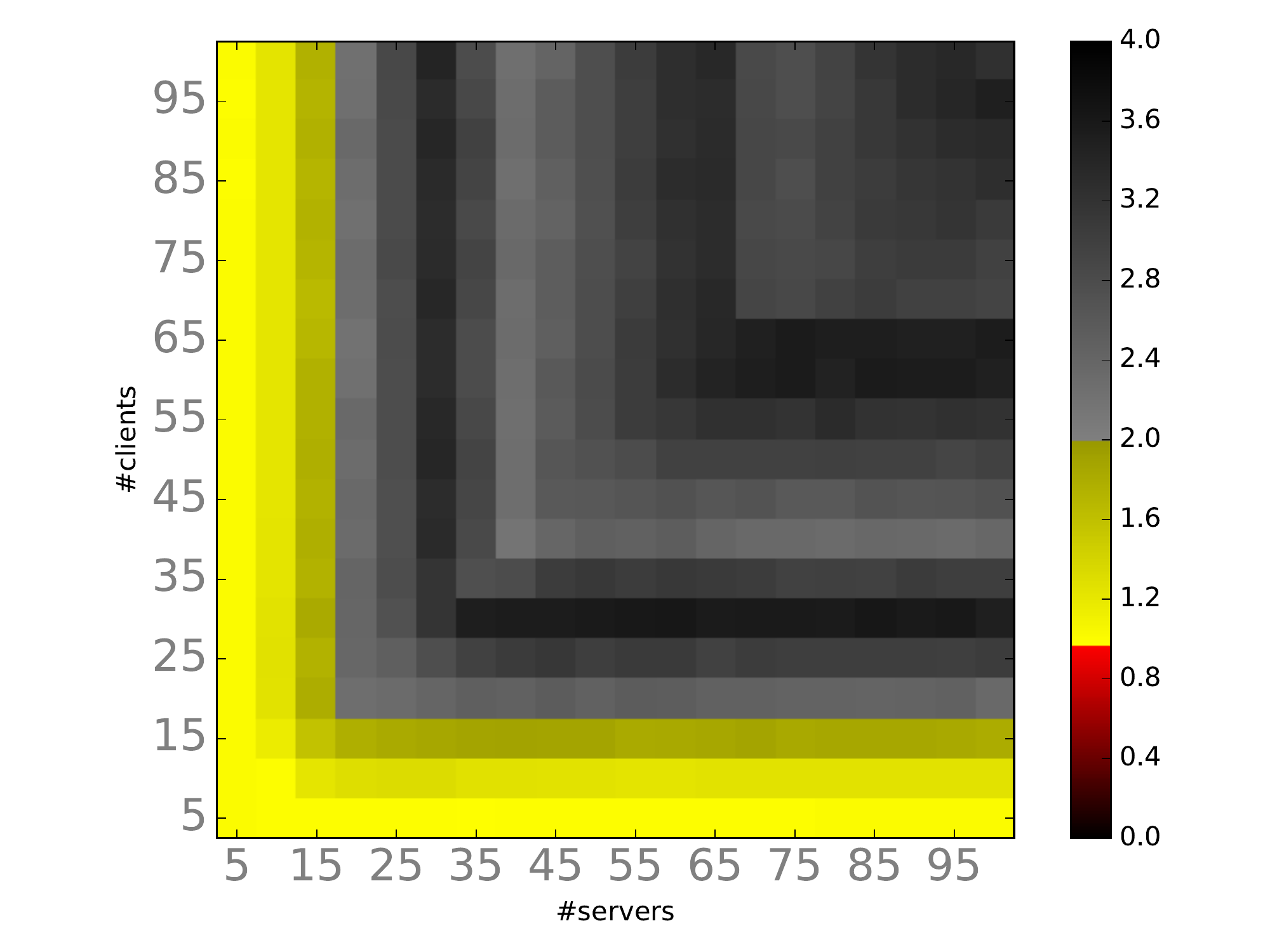}
      \caption{Small runs K Disjoint paths, OS=4}
     \label{fig:kdisjoint cs small}   
    \end{subfigure}
  \end{minipage}
    \begin{minipage}[b]{0.32\textwidth}
    \centering
    \hskip -0.1 cm
    \begin{subfigure}[b]{\linewidth}
        \includegraphics[width=0.95\textwidth]{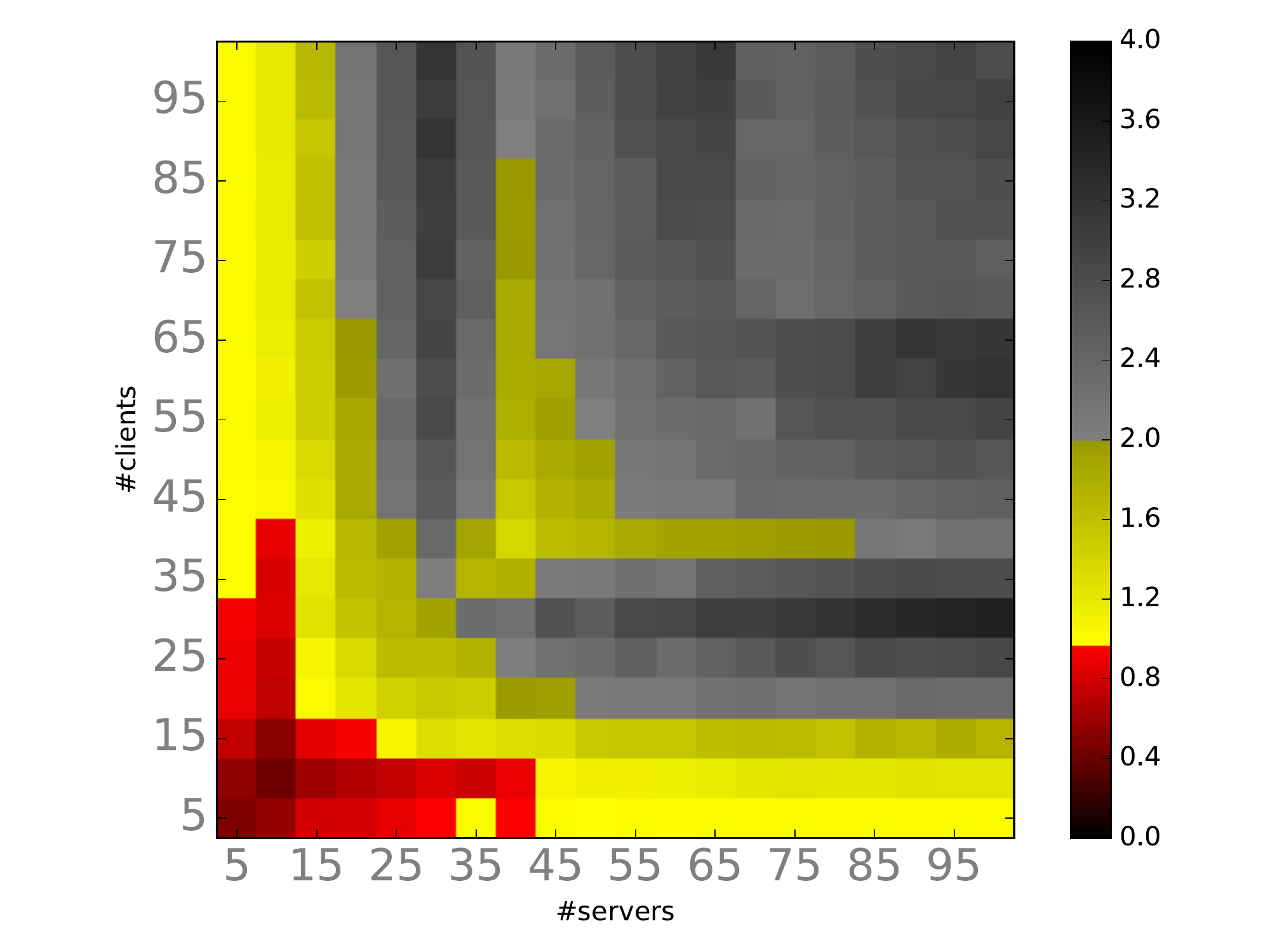}
        \caption{Small runs ECMP, OS=4}
        \label{fig:ecmp cs small}
    \end{subfigure}
    \begin{subfigure}[b]{\linewidth}
        \includegraphics[width=0.95\textwidth]{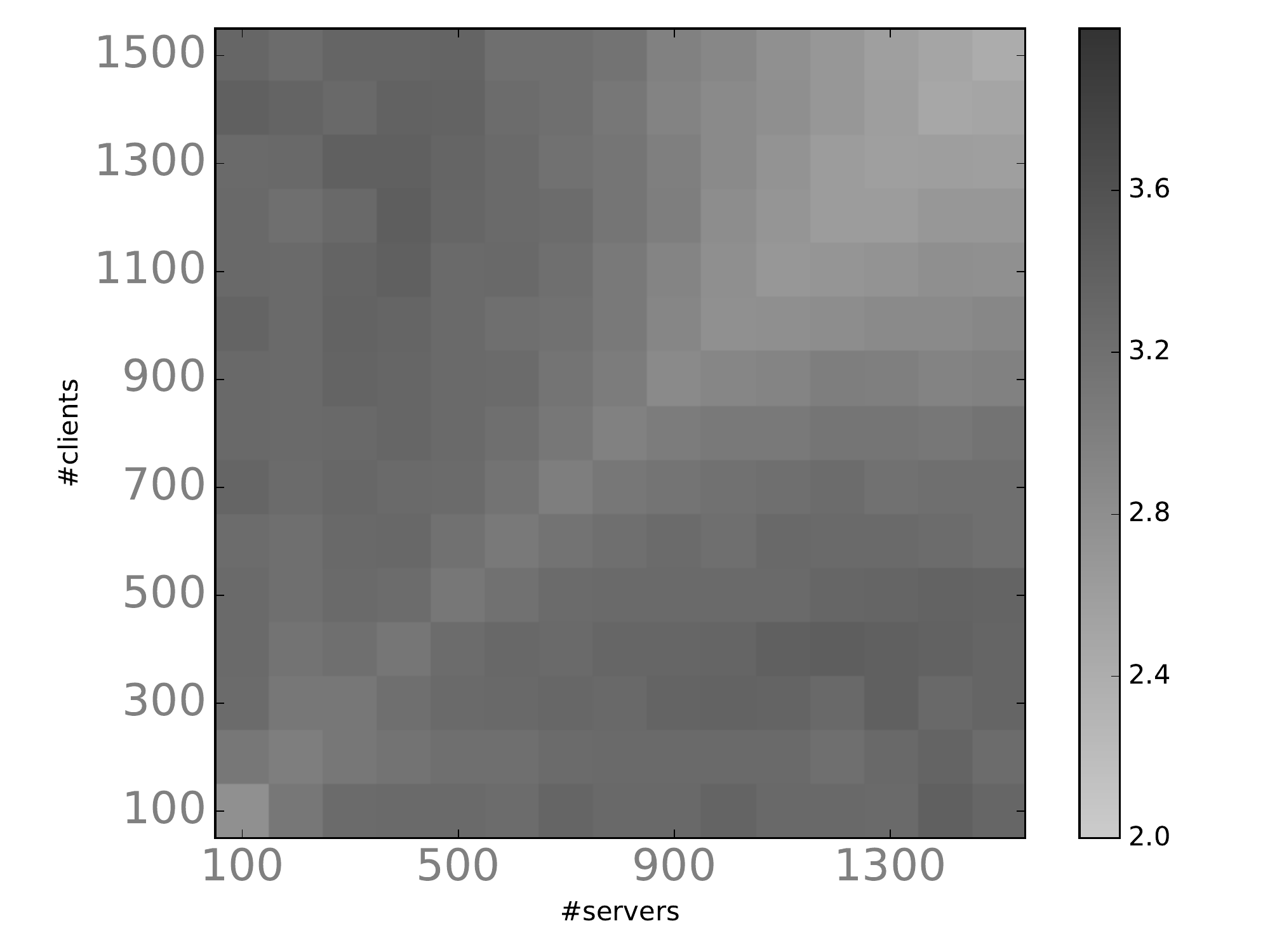}
        \caption{Large runs ECMP, OS=4}
        \label{fig:ecmp cs large}
    \end{subfigure}
  \end{minipage}
  \label{fig:csJfFat}
  \compactcaption{Random Graph vs Fat tree $C-S$ model: Different routing schemes on small and large runs for the C-S model with different oversubscription ratios. Each tile in the heatmap is the average throughput per flow in the random graph divided by that in the fat tree for the given value of $C$ and $S$}
\end{figure*}

\subsection{Routing and Transport}

We used shortest path routing with ECMP and TCP in the transport layer for all sets of experiments. Additionally, for bandwidth experiments, we also evaluate k-shortest path routing~\cite{kShortestPaths} and k edge-disjoint path routing (henceforth, referred to as k-disjoint path routing) to contrast with shortest path routing. In Section~\ref{kshortimpl} we discuss a simple practical variant of k-disjoint path routing that can be easily implemented with segment routing.

\subsection{C-S Model} \label{cs model}
Traffic patterns in a datacenter have high variance and change rapidly with time~\cite{ghobadi16}. Morever, there is no underlying theme across datacenters~\cite{roy2015,benson2010,delimitrou2012}. Majority of the traffic was not found to be rack-local in Facebook's datacenter~\cite{roy2015} , in contrast to findings in ~\cite{benson2010,delimitrou2012}. \cite{ghobadi16} showed that traffic pattern is highly skewed in Microsoft's datacenters with only a very little fraction of racks carrying most of the traffic, while in Facebook's datacenter~\cite{roy2015}, demand was wide-spread, uniform and stable across a longer time period.

Motivated by these arguments, we believe that there is a need to succinctly represent a wide range of workloads for topology evaluations. Previous works have used permutation matrices~\cite{singla12}, all to all patterns~\cite{singla12,guo2008}, near worst case longest matchings~\cite{jyothi2016}, skewed traffic patterns~\cite{ghobadi16,kassing17} and snapshots from real  datacenters. While all the above traffic matrices have different characteristics (worst case, skewed or  uniform) and provide useful insights, we argue that these are very specific experimental points in a wide range of possibilities and hence do not portray the entire picture.

We propose a simple model to capture a wide range of scenarios. We pick a subset of client hosts, say set $C$ and pack these clients into the fewest number of racks while randomly choosing the racks in the datacenter. Similarly, we pick a subset of server hosts, say set $S$ and pack all servers into the fewest number of racks possible (avoiding racks that have been used for $C$). We wish to measure the network capacity between the sets $C$ and $S$ for all possible sizes of $C$ and $S$. 
By varying the sizes of $C$ and $S$, we claim that this model, which we call the $C-S$ model, captures a wide range of patterns that commonly occur in applications.

\begin{itemize}[leftmargin=*]
\itemsep0em
\item Setting $|S|$ = 1 results in incast pattern. Similarly, setting $|C|$ = 1 can mimic outcasts. In Section~\ref{extreme case analysis}, we simulate different types of incast and outcast scenarios with specific parameter settings of $C$ and $S$.
\item If $C$ represented the map tasks and $S$ all reduce tasks of a map-reduce job, then the $C-S$ model can capture the shuffle operation between the map and reduce tasks.
\item If $C$ represented the parameter servers and $S$ the workers for ML training, the  $C-S$ model can capture model updates from workers to the parameter servers and parameter reads from parameter servers back to the workers.
\item Setting $|S| = |C| = r$, where $r$ is the number of servers in a rack, captures rack-to-rack traffic pattern.  
\item With $|C| << |S|$, the $C-S$ model captures a wide range of skewed traffic patterns. 
\item Uniform traffic patterns across the entire datacenter by setting $|C| = |S| = n/2$, where $n$ is the total number of hosts. 
\end{itemize}



\section{Bandwidth capacity} \label{bandwidth capacity}
We study bandwidth capacity in the $C-S$ model. By covering a wide range of scenarios, the $C-S$ model allows us to identify scenarios where expanders have an advantage over fat-trees and leaf spines, and where it doesn't.
We graph the $C-S$ model using a heatmap, varying the size of sets $C$ and $S$ on the x and y-axis. Between every client in $C$ to every server in $S$, we run a long running TCP connection with infinite data to send. Each cell in the heat map represents the average throughput per flow in the random graph normalized by the average throughput in the baseline topology. We plot two heatmaps in each case for (a) small runs comprising of smaller values of $C$ and $S$, upto a few handful racks and (b) larger runs comprising of larger values of $C$ and $S$, almost all the way till the entire datacenter. We used the simulator based on htsim from~\cite{singla12} for these experiments.

\subsection{Comparisons with Fat tree} \label{bandwidth capacity fat tree}
As noted in ~\cite{kassing17}, a lot of bandwidth capacity remains unused in a fat tree when only a few racks are sending traffic. This is also evident from the heatmaps (Figures~\ref{fig:ecmp cs small os1}-~\ref{fig:ecmp cs large}). For a large range of $C$ and $S$, the expander gets drastically more throughput than the fat tree (3-4x in many cases, see Figure~\ref{fig:ecmp cs large}).


\subsubsection{Effect of oversubscription}
Figures~\ref{fig:ecmp cs small os1}, \ref{fig:ecmp cs small os2} and \ref{fig:ecmp cs small} illustrate the effect of oversubscription for small sizes of $C$ and $S$ in the $C-S$ model. With an oversubscription ratio of 4, for smaller runs, the expander drastically outperforms the fat tree for all sizes of $C$ and $S$ (except for the red patch towards the bottom left). For runs (Figure~\ref{fig:ecmp cs large}) with larger values of $C$ and $S$, the expander network outperforms fat trees drastically in all cases. 
Also observe that when run with oversubscription 1 (Figure~\ref{fig:ecmp cs small os1}), the difference between expander and fat trees are unapparent because the network is not the bottleneck, highlighting problems of comparisons with full bandwidth. The differences get increasingly evident as the oversubscription ratio is increased as effects of rack oversubscription come into the picture. Expanders significantly helps out the hot racks by alleviating rack oversubscription.

An oversubscribed topology also brings out cases where expander graphs do not perform well (the red patches on bottom left of Figure~\ref{fig:ecmp cs small}). As noted in ~\cite{kassing17}, expander being a flat topology (all switches connect to servers) offers fewer shortest paths between a pair of nodes and hence when two racks send and receive data at full link speed, the expander network does not match fat tree's performance with shortest path routing. This is the only region in the $C-S$ model where expander with shortest path routing does not match the performance of a fat tree. 

\subsubsection{Effect of routing scheme}
Figures ~\ref{fig:ecmp cs small}, \ref{fig:kshort cs small} and \ref{fig:kdisjoint cs small} illustrate the performance of random graphs with different routing strategies namely shortest paths, k-shortest paths and k-disjoint paths. As shown in Figure~\ref{fig:kdisjoint cs small}, k-disjoint path routing eliminates the few cases where the random network is not able to match fat tree's performance with shortest paths. One of the reasons is that k-disjoint routing distributes load evenly just at the edge of the network (network links at the ToR).  Although naively k-disjoint paths require perfect source routing, we show that they can be approximated very accurately by combinations of two shortest paths, which can be easily implemented with MPLS or segment routing (Section~\ref{kshortimpl}).
We emphasize that expander networks are already advantageous with ECMP and shortest paths in most cases. K-disjoint path routing might not be necessary for practical purposes.

\begin{figure}
\centering
\includegraphics[width=0.4\textwidth]{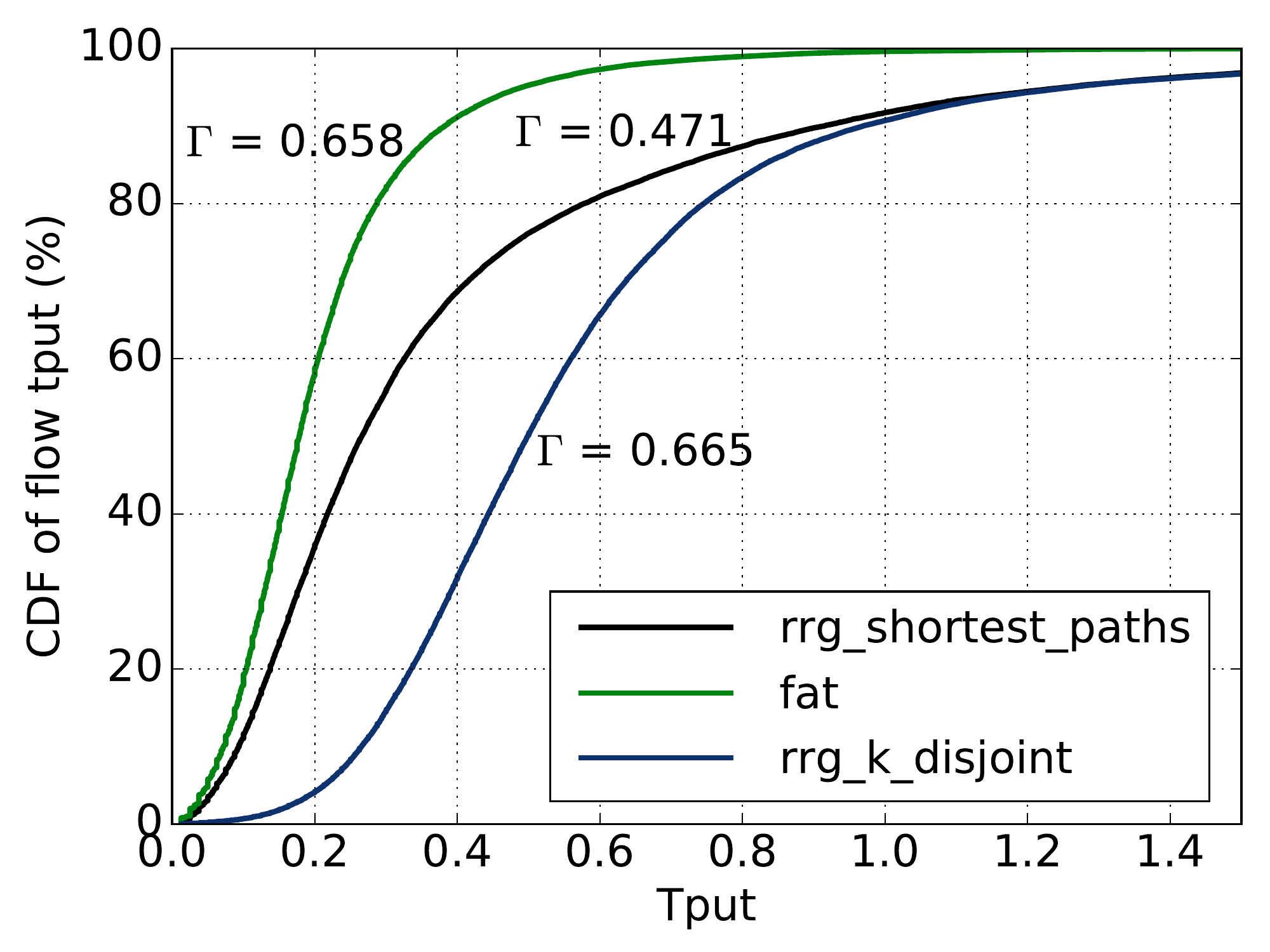}
\compactcaption{CDF of flow throughput for C=200, S=200, annotated with their corresponding Jain's fairness index ($\Gamma$).}
\label{fig:fairness}
\end{figure}

\begin{figure*}
  \vskip -0.8 cm
  \begin{minipage}[b]{0.24\textwidth}
    \centering
    \begin{subfigure}[b]{\linewidth}
        \includegraphics[width=1.05\textwidth]{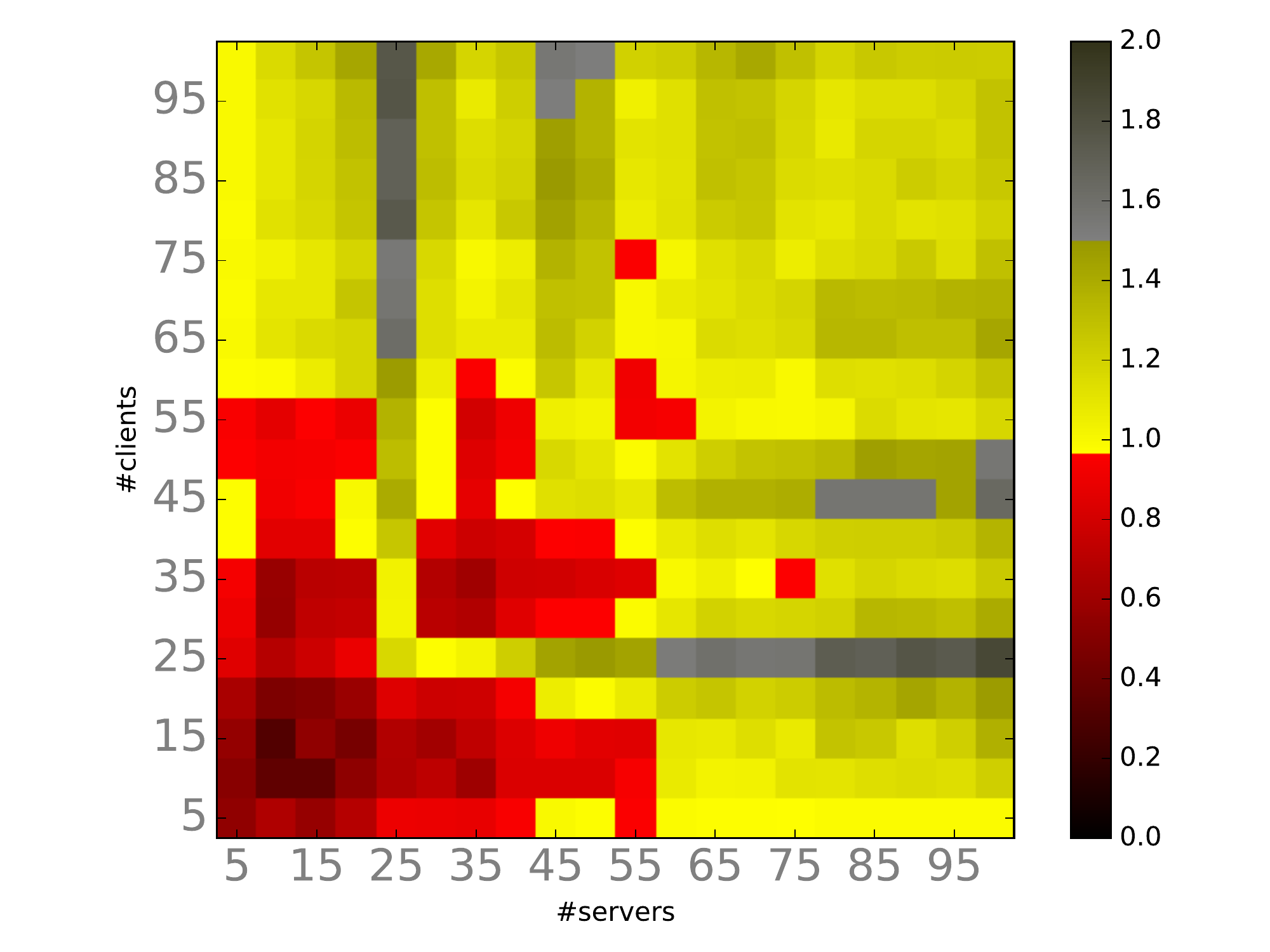}
        \caption{Small runs ECMP}
        \label{fig:ecmp ls small} 
    \end{subfigure}
  \end{minipage}
  \begin{minipage}[b]{0.24\textwidth}
    \centering
    \hskip -0.1 cm
  \begin{subfigure}[b]{\linewidth}
    \includegraphics[width=1.05\textwidth]{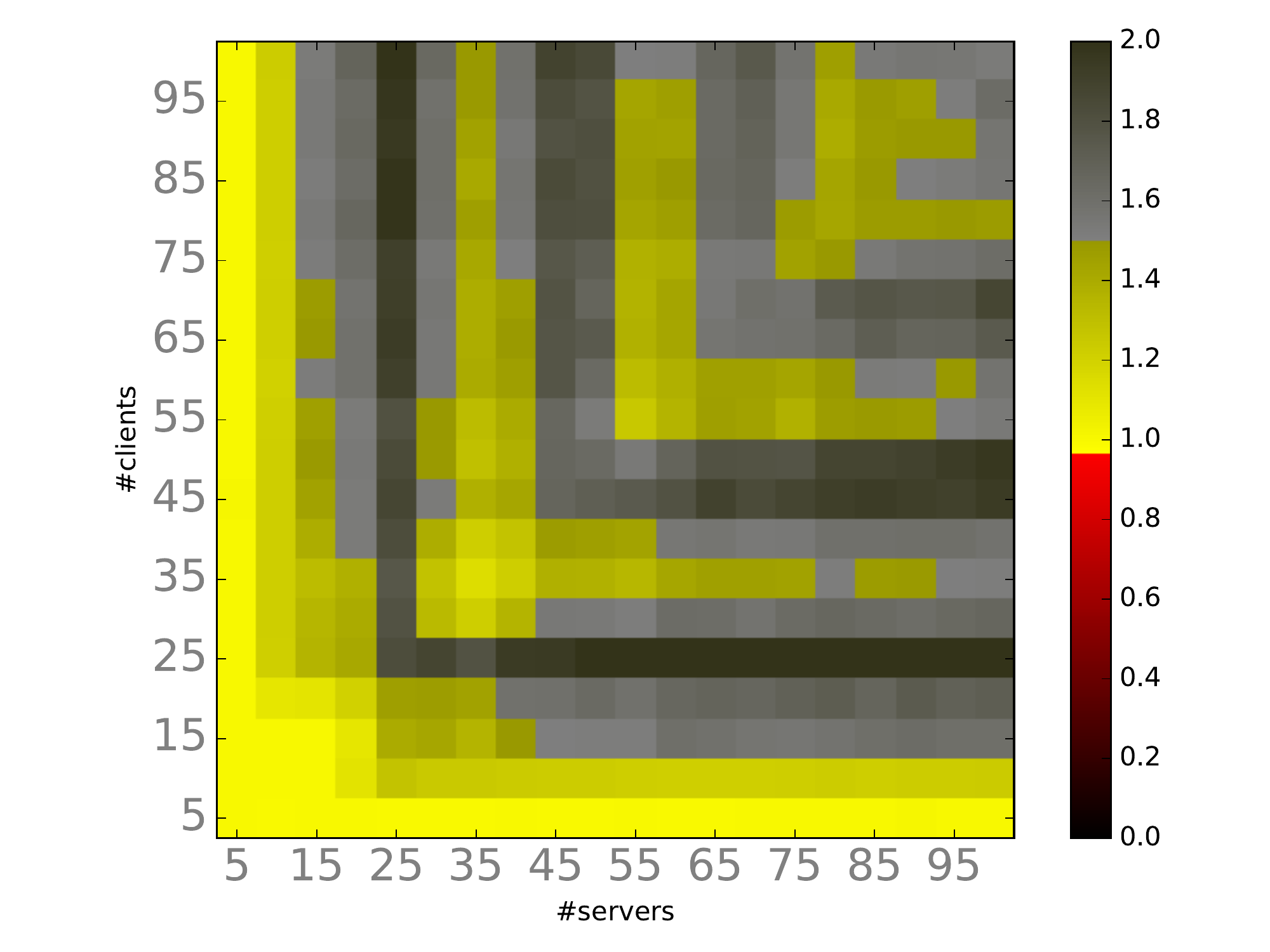}
    \caption{Small runs K Disjoint}
    \label{fig:kdisjoint ls small}
  \end{subfigure}
  \end{minipage}
  \begin{minipage}[b]{0.24\textwidth}
    \centering
    \hskip -0.1 cm
    \begin{subfigure}[b]{\linewidth}
        \includegraphics[width=1.05\textwidth]{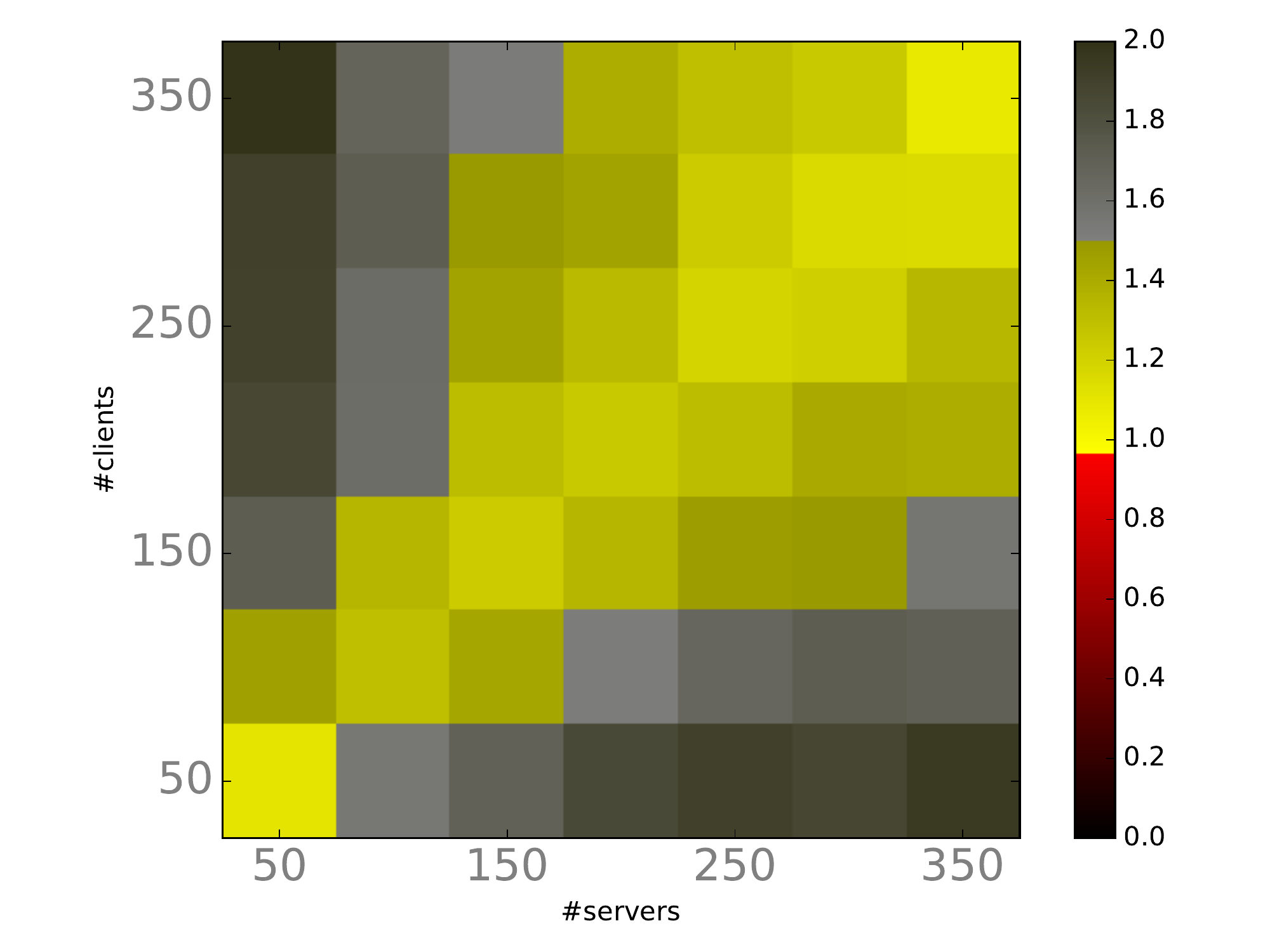}
        \caption{Large runs ECMP}
        \label{fig:ecmp ls large}  
    \end{subfigure}
  \end{minipage}
   \begin{minipage}[b]{0.24\textwidth}
    \centering
    \begin{subfigure}[b]{\linewidth}
    \includegraphics[width=1.05\textwidth]{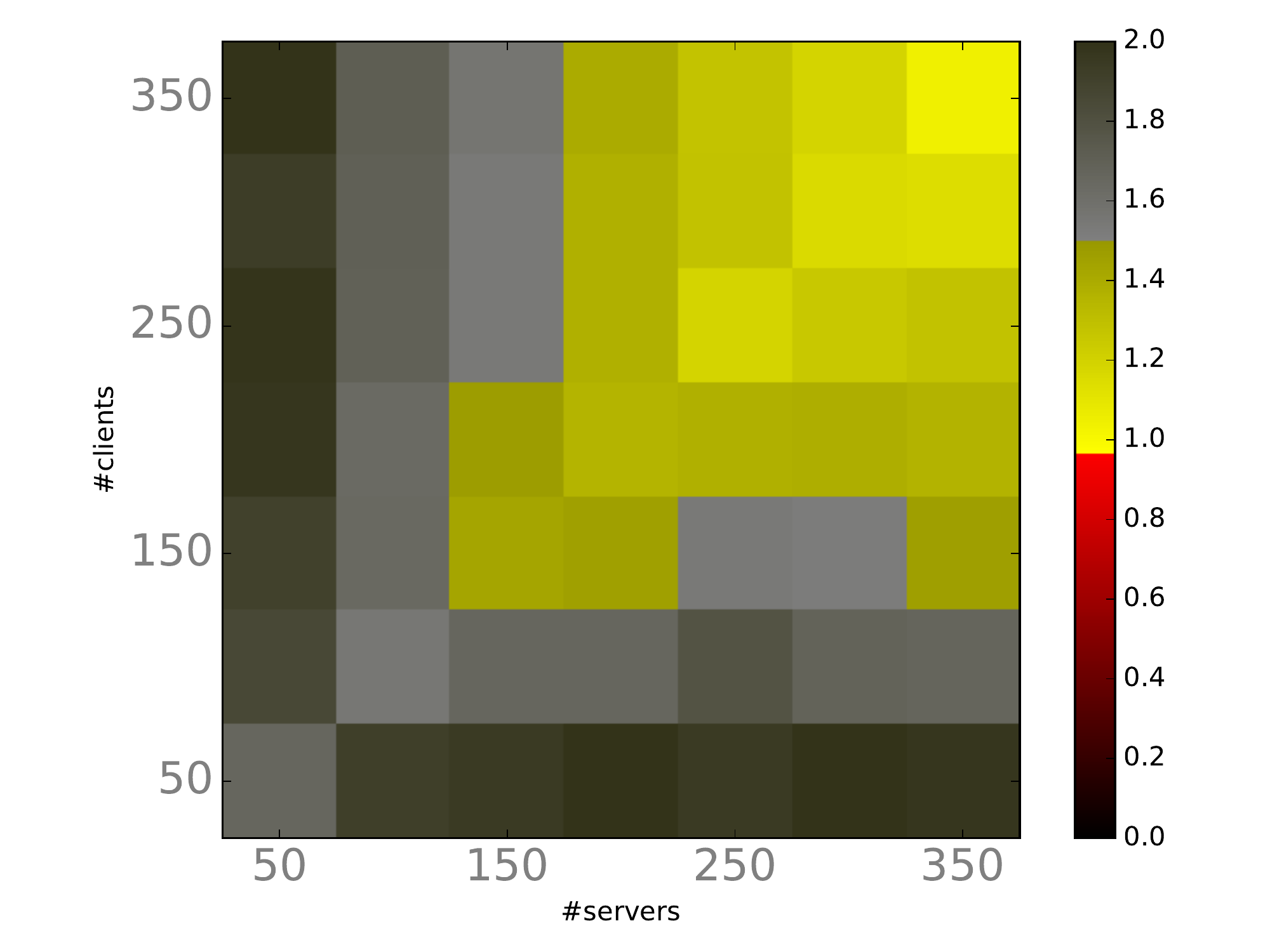}
    \caption{Large runs K Disjoint}
    \label{fig:kdisjoint ls large}
  \end{subfigure}
  \end{minipage}
  \compactcaption{Jellyfish vs Leaf Spine in the C-S model: Different routing schemes on small and large runs.}
    \label{fig:csJfLs}
\end{figure*}

\vskip 0.4 cm
The drastic improvement can be partly attributed to lesser rack oversubscription in an expander (more exit network ports per server in a rack, true for any flat network). Note that unlike a fat tree, the network links of a ToR switch in a flat network carry traffic that originated in the rack as well as traffic that originated elsewhere but routed through that ToR  switch. Nevertheless, all network links can carry local outgoing traffic in the best case. This is especially true for micro bursts where a rack has a lot of traffic to send for a short period of time and traffic is well multiplexed at the network links (very few racks are bursting at any given point). Motivated by these arguments, we introduce the notion of \textbf{Uplink to Downlink Factor} or \textbf{UDF} of a topology. Consider a topology $T$ and a random graph topology $R(T)$ built with the same equipment. For every top of the rack switch that contains servers, we look at the ratio of network ports and the number of server ports call this ratio \textbf{NSR} (for Network Server Ratio). 
We define UDF($T$) as, \[UDF(T) = \frac{NSR(R(T))}{NSR(T)}\]

Intuitively, NSR represents the outgoing network capacity per server in a rack. The UDF represents the best case scenario for the random graph when the outgoing links carry only traffic originated from the rack. It represents an upper bound on the performance of random graph as compared to the baseline topology.

\begin{figure}
\centering
\includegraphics[width=0.43\textwidth]{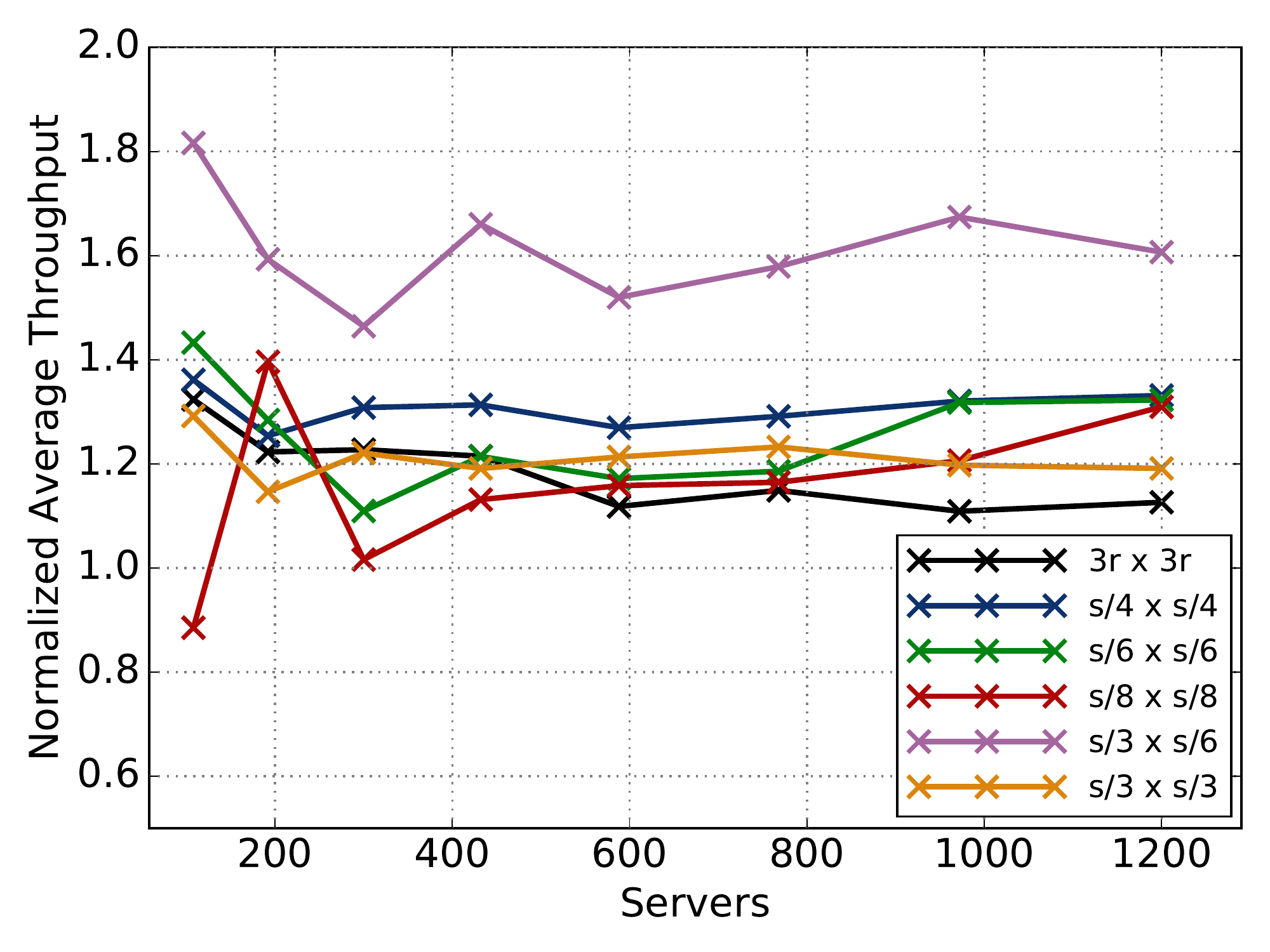}
\compactcaption{Effect of scale: Average throughput in a random graph normalized by the average throughput in a leaf spine. Each line represents a specific point in the $C-S$ model. $r$ represents the number of servers in one rack and $s$ represents the total number of servers in the datacenter. On the x-axis, is the network size corresponding to a leaf spine with oversubscription ratio = 3, supporting that many servers.}
\label{fig:scale experiment}
\end{figure}

We can easily compute the UDF for a fat tree with switch degree $k$.
\begin{align*}
UDF\big(T = FatTree(k)\big) &= \frac{NSR(R(T))}{NSR(T)} \\
&= \frac{\Big(k - \big(\frac{k^3}{4}/\frac{5k^2}{4}\big)\Big)\Big/\big(\frac{k^3}{4}/\frac{5k^2}{4}\big)}{\frac{k}{2}\big/\frac{k}{2}} \\ 
&= 4
\end{align*}

Note that number of server ports in one rack of a random graph built with the same equipment as a Fat tree with switch degree = $k$ is $\big(\frac{k^3}{4}/\frac{5k^2}{4}\big)$ and the number of network ports in one rack is equal to $(k - \text{server ports})$. From Figures~\ref{fig:ecmp cs small} and ~\ref{fig:ecmp cs large}, random graphs achieve throughput that is very close to UDF: 4 in the case of fat trees.

\subsubsection{Note on Fairness}
We note that the differences in average throughput, highlighted in Section~\ref{bandwidth capacity fat tree}, with higher oversubscription ratios are not because the routing inefficiencies get masked by the presence of a large number of flows. If that were the case, some flows would utilize all existing capacity and the throughput distribution would be unfair.   Figure~\ref{fig:fairness} plots the CDF of throughputs of flow for $|C| = 200, |S| = 200$. It can be seen that Jain's Fairness index of random graphs matches that of fat tree with k-disjoint path routing. With shortest path-routing, the throughput distribution of flows is strictly better for the random graph than the fat tree.

\subsection{Comparisons with Leaf Spine} \label{leaf spine comparisons}

\begin{figure*}
  \vskip -0.8 cm
  \begin{minipage}[b]{0.32\textwidth}
    \centering
    \begin{subfigure}[b]{\linewidth}
        \includegraphics[width=1.00\textwidth]{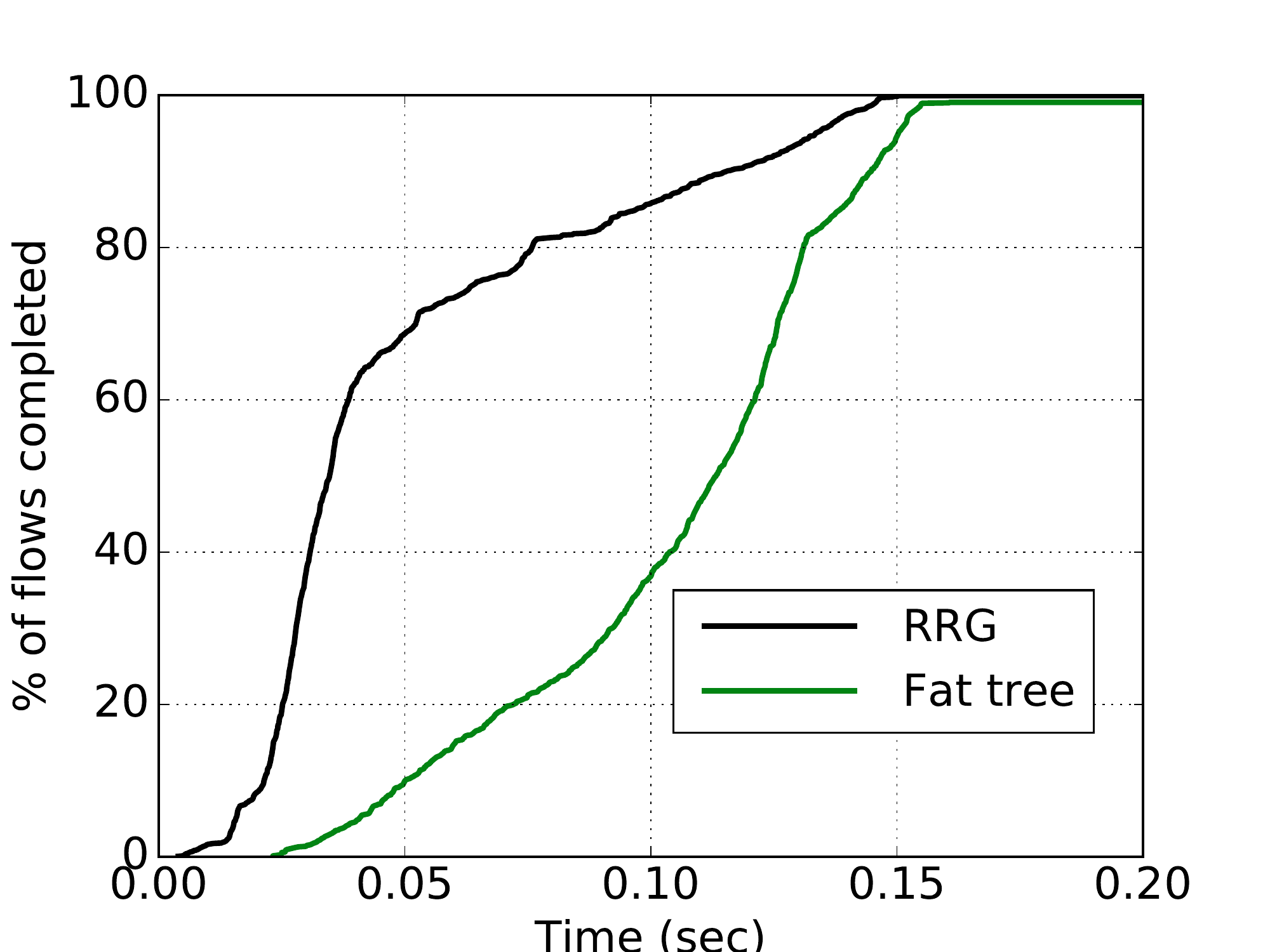}
        \caption{Flow completion time distribution}
        \label{fig:incast flow dist} 
    \end{subfigure}
  \end{minipage}
  \begin{minipage}[b]{0.33\textwidth}
    \centering
    \begin{subfigure}[b]{\linewidth}
    \includegraphics[width=1.00\textwidth]{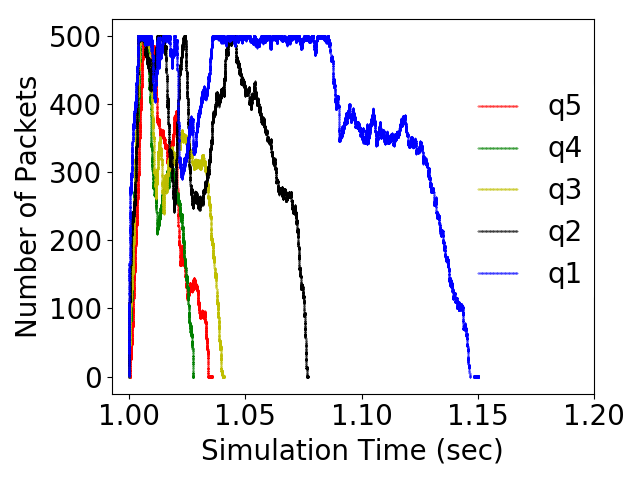}
    \caption{Random graph: 5 busiest queues over time}
    \label{fig:queue buildup rrg}
  \end{subfigure}
  \end{minipage}
\begin{minipage}[b]{0.33\textwidth}
    \centering
    \begin{subfigure}[b]{\linewidth}
    \includegraphics[width=1.00\textwidth]{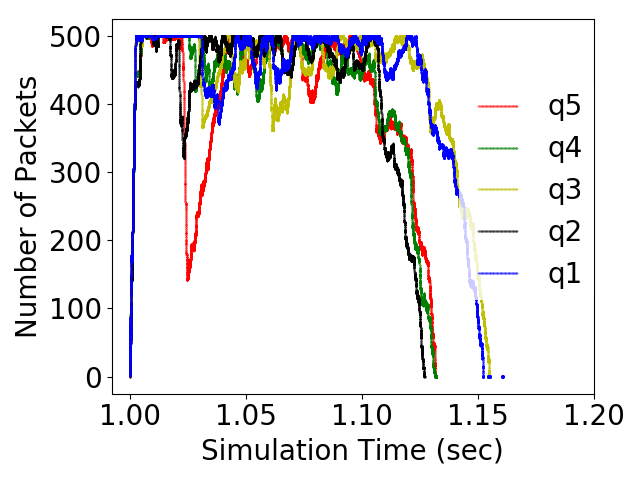}
    \caption{Fat tree: 5 busiest queues over time}
    \label{fig:queue buildup fat}
  \end{subfigure}
  \end{minipage}
  \caption{40:20 Incast experiment to reproduce incast effects due to ToR fan-in at the aggregation layer~\cite{JupiterRising15}. These incasts are a result of many-to-one rack level traffic patterns (in contrast to many-to-one server level incast patterns).}
    \label{fig:incastOutcastFlowDist}
\end{figure*}

\begin{figure*}
  \begin{minipage}[b]{0.33\textwidth}
    \centering
    \begin{subfigure}[b]{\linewidth}
        \includegraphics[width=\textwidth]{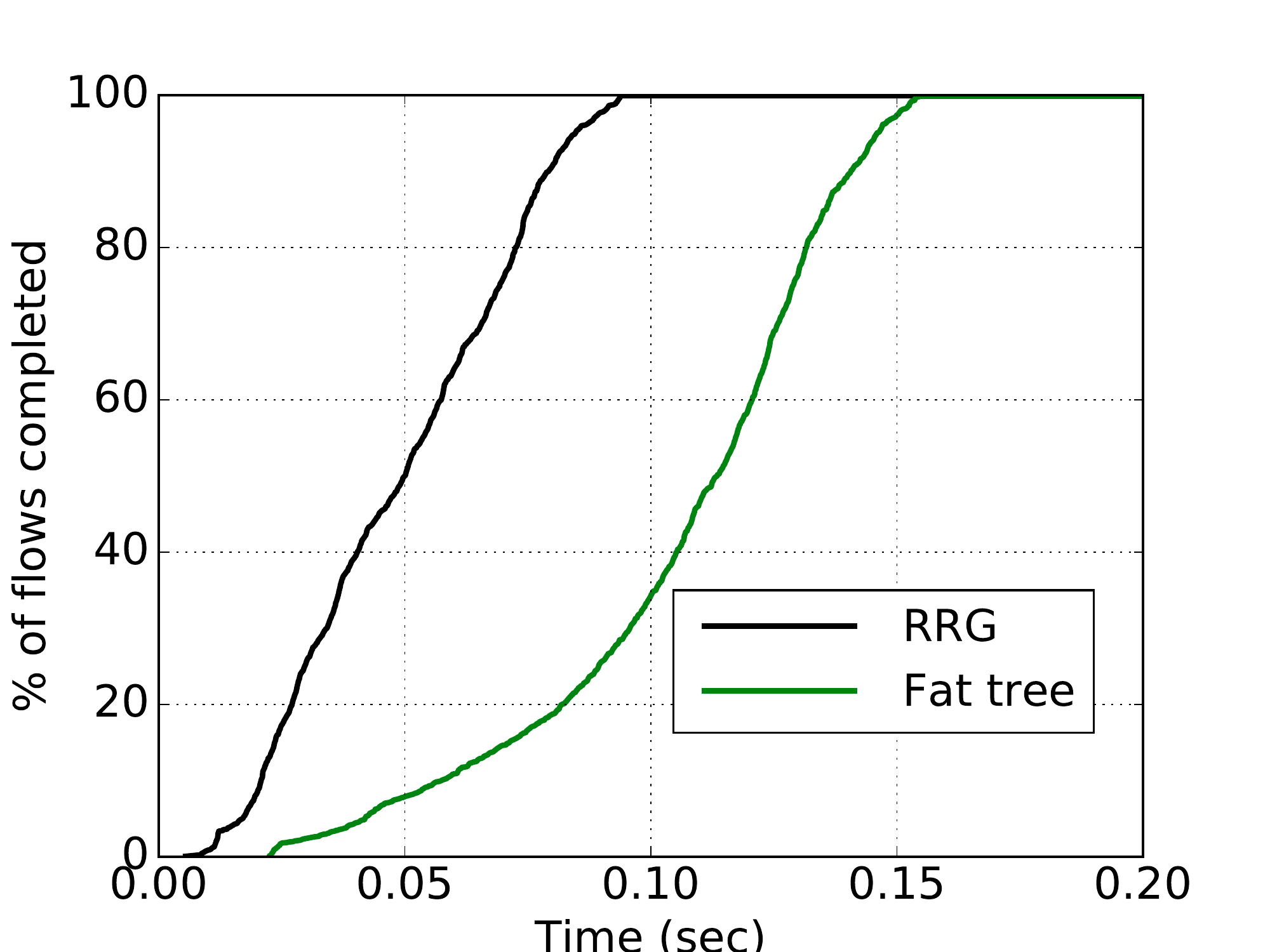}
        \caption{Flow completion time distribution}
        \label{fig:outcast flow dist} 
    \end{subfigure}
  \end{minipage}
  \begin{minipage}[b]{0.33\textwidth}
    \centering
    \begin{subfigure}[b]{\linewidth}
    \includegraphics[width=\textwidth]{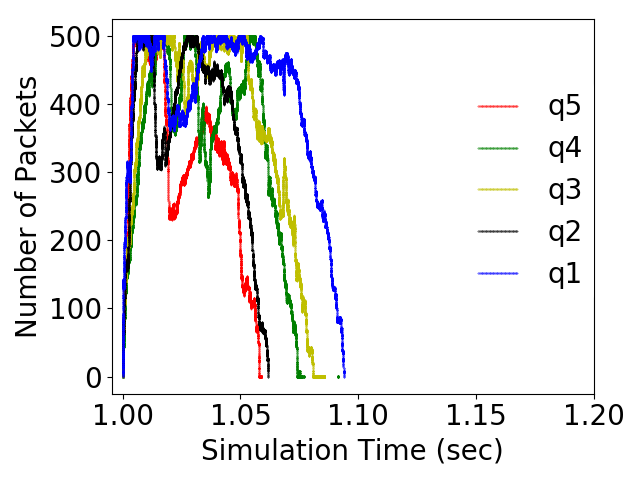}
    \caption{Random graph: 5 busiest queues over time}
    \label{fig:queue buildup rrg outcast}
  \end{subfigure}
  \end{minipage}
\begin{minipage}[b]{0.33\textwidth}
    \centering
    \begin{subfigure}[b]{\linewidth}
    \includegraphics[width=\textwidth]{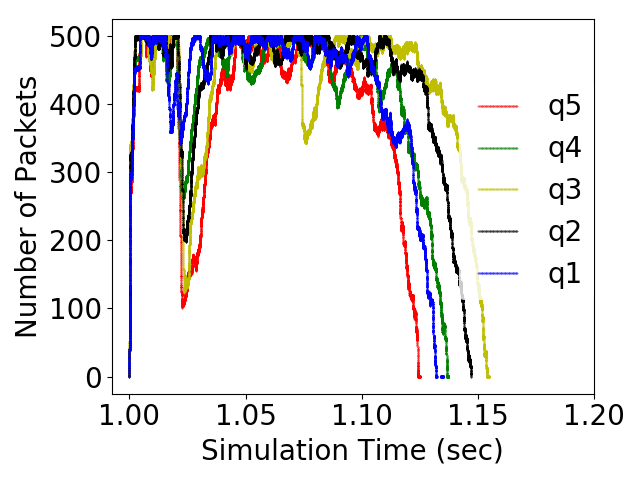}
    \caption{Fat tree: 5 busiest queues over time}
    \label{fig:queue buildup fat outcast}
  \end{subfigure}
  \end{minipage}
  \compactcaption{20:40 Outcast experiment to reproduce outcast effects due to oversubscription of ToR uplinks~\cite{JupiterRising15}}
    \label{fig:incastOutcastFlowDist}
\end{figure*}

Leaf spine topologies are popular for small to medium tiered datacenters consisting of a few tens or hundreds of racks~\cite{AristaBestPractices}. Leaf spine topologies are two-tiered tree based topologies where all lower layer ToR switches (leaves) are connected to all upper layer switches (spines). The servers are distributed across all the leaf switches. The topology provides multiple paths between any source and destination (one path through each spine) and is meant to provide good load balance.

In this section we compare leaf spine topologies to random graphs built with the same equipment. 
Figure~\ref{fig:csJfLs} illustrates the throughput in the $C-S$ model for small and large runs with $x=24$ and $y=8$. Each tile is normalized by the throughput for the leaf spine for that $C$ and $S$. As with fat trees, with ECMP, random graph outperforms the leaf spine topology in all regions except the bottom left red patch of Figure~\ref{fig:ecmp ls small} representing the bandwidth between two racks.

We can compute the UDF for leaf spine switches for arbitrary $x$ and $y$.
We have, \begin{align*}
NSR(T=LeafSpine(x, y)) &= \frac{y}{x}
\end{align*}
For the corresponding random graph $R(T)$ built with the same equipment,
\begin{align*}
NSR(R(T)) &= \frac{(x+y) - \text{server ports per switch}}{\text{server ports per switch}}\\
&= \frac{(x+y) - \big(x(x+y)/(x+2y))}{x(x+y)/(x+2y)}\\
&= \frac{2y}{x} \\
\text{Thus, } \quad UDF\big(T&=LeafSpine(x, y)\big) = \frac{NSR(R(T))}{NSR(T)} = 2 
\end{align*}
It can be seen from Figures~\ref{fig:ecmp ls small}, ~\ref{fig:kdisjoint ls large} that expanders indeed touch the upper bound of $UDF=2$, against leaf spines, in the $C-S$ model. Observe that the $UDF$ of a leaf spine network is independent of the number leaf and spine switches. If a network has fewer spines and more leaves, the number of servers per rack are fewer but the aggregate uplink bandwidth at the ToRs is also lesser. These two factors cancel each other and hence, the UDF remains constant.

\begin{figure*}
 \vskip -0.1 cm
    \begin{minipage}[b]{0.28\textwidth}
    \centering
    \hskip 0.5 cm
    \includegraphics[width=0.9\textwidth]{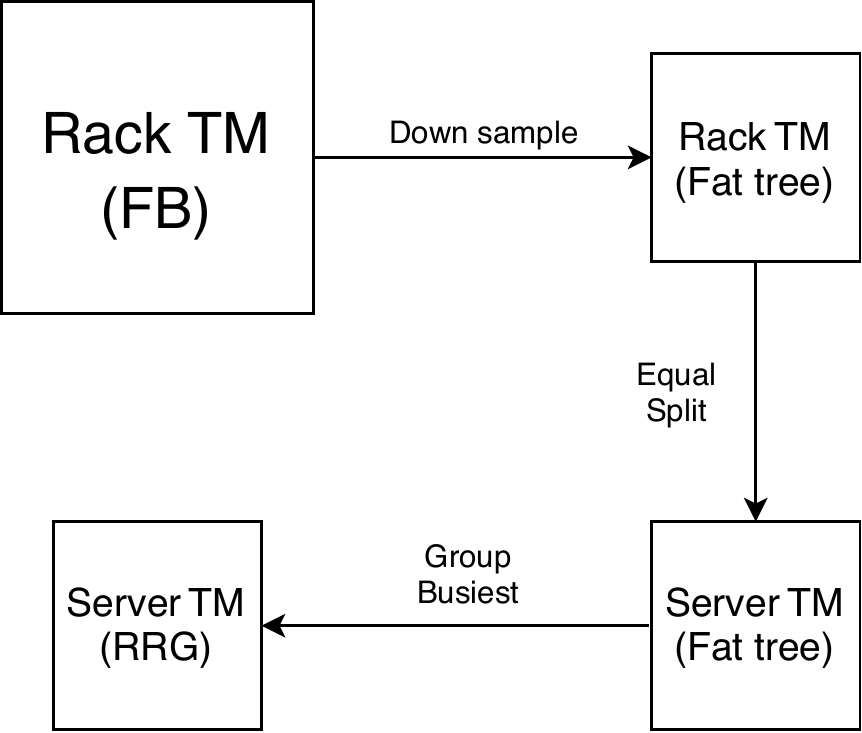}
    \compactcaption{FB trace downsampling method}
    \label{fig:fb method}
  \end{minipage}
  \hskip 0.5 cm
  \begin{minipage}[b]{0.3\textwidth}
    \centering
  \includegraphics[width=\textwidth]{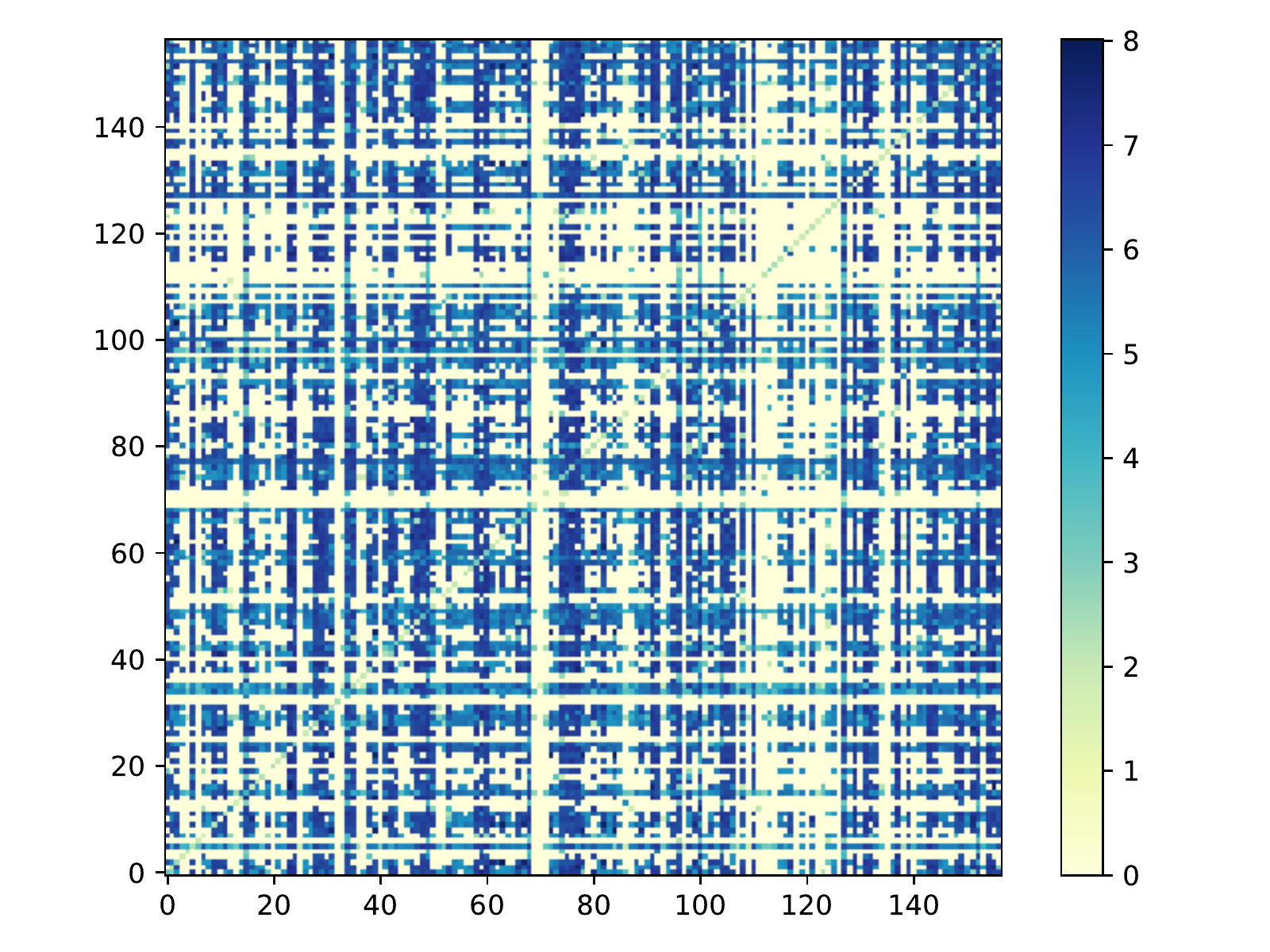}
  \compactcaption{Rack level traffic matrix for the chosen FB cluster (log base 10)}
  \label{fig:fb rack tm}
  \end{minipage}
  \hskip 0.5 cm
   \begin{minipage}[b]{0.31\textwidth}
    \centering
    \includegraphics[width=0.9\textwidth]{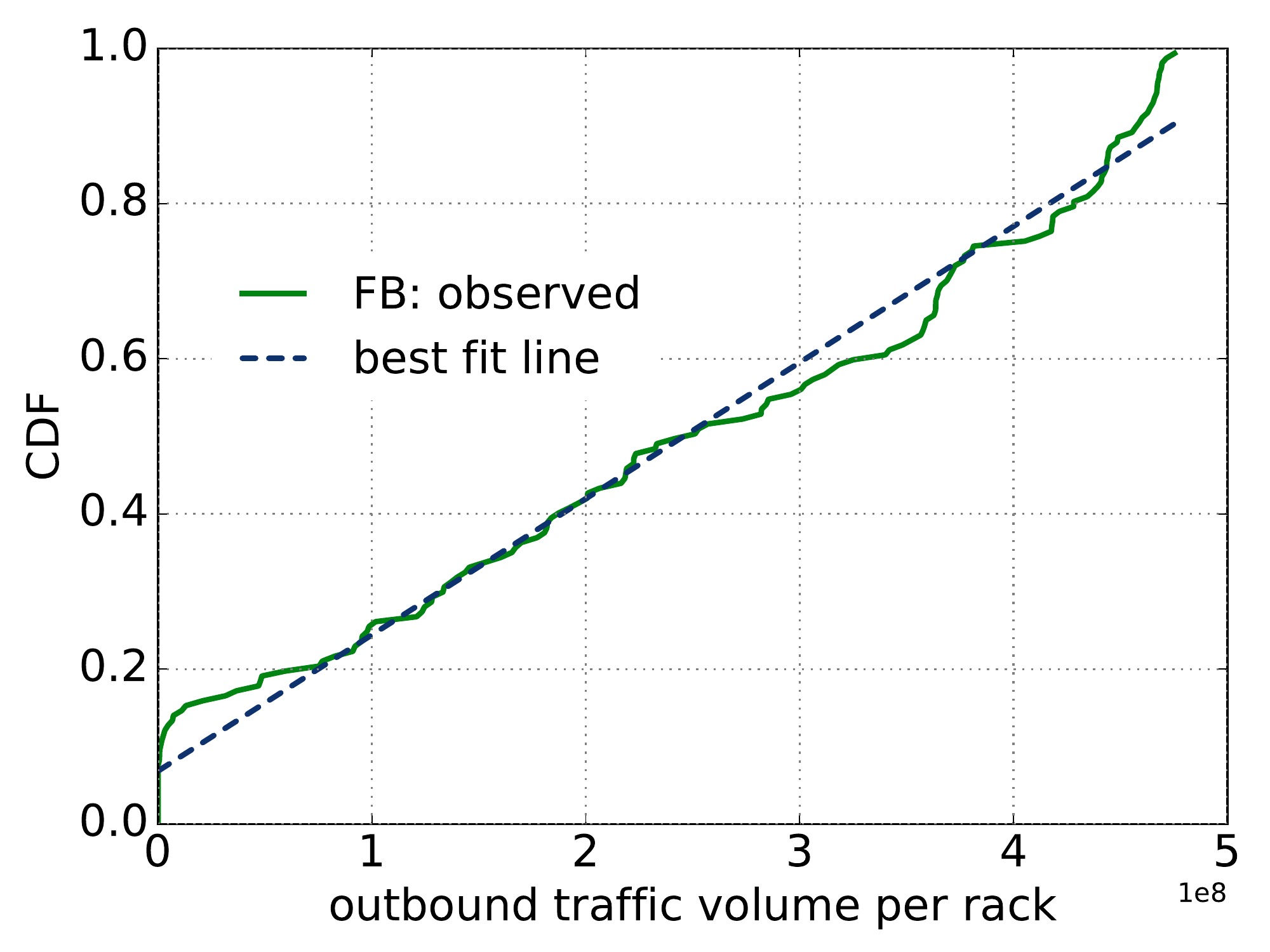}
    \compactcaption{The outgoing volume of traffic from a rack in the FB cluster follows a near-uniform distribution.}
    \label{fig:fb distr}
  \end{minipage} 
\end{figure*}

\begin{figure*}
  \hskip -0.1 cm
\begin{minipage}[b]{0.3\textwidth}
    \centering
    \begin{subfigure}[b]{\linewidth}
    \includegraphics[width=\textwidth]{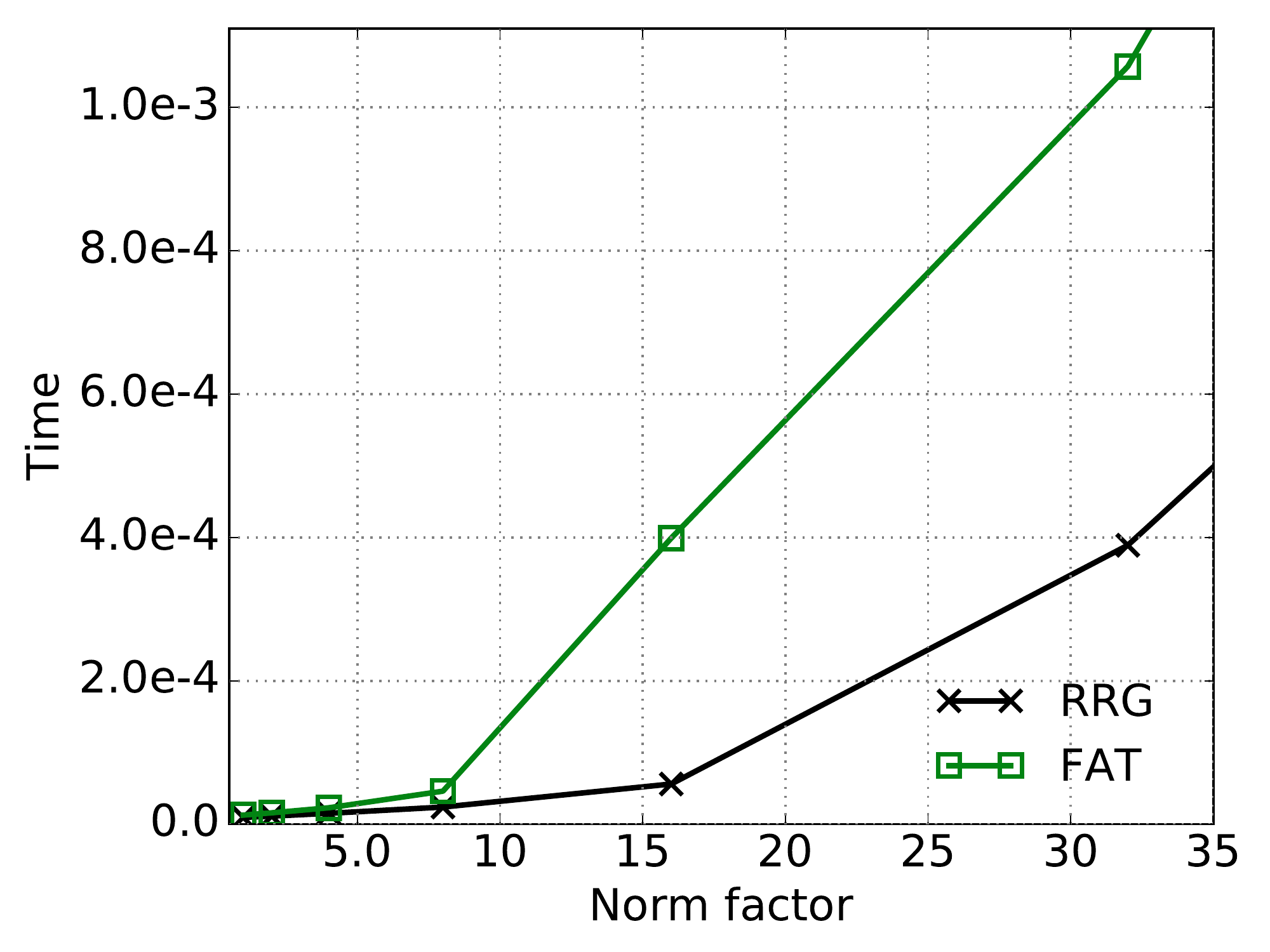}
    \compactcaption{50 \%ile FCT}
    \label{fig:fb 50}
  \end{subfigure}
  \end{minipage}
  \hskip 0.25cm
   \begin{minipage}[b]{0.3\textwidth}
    \centering
    \begin{subfigure}[b]{\linewidth}
    \includegraphics[width=\textwidth]{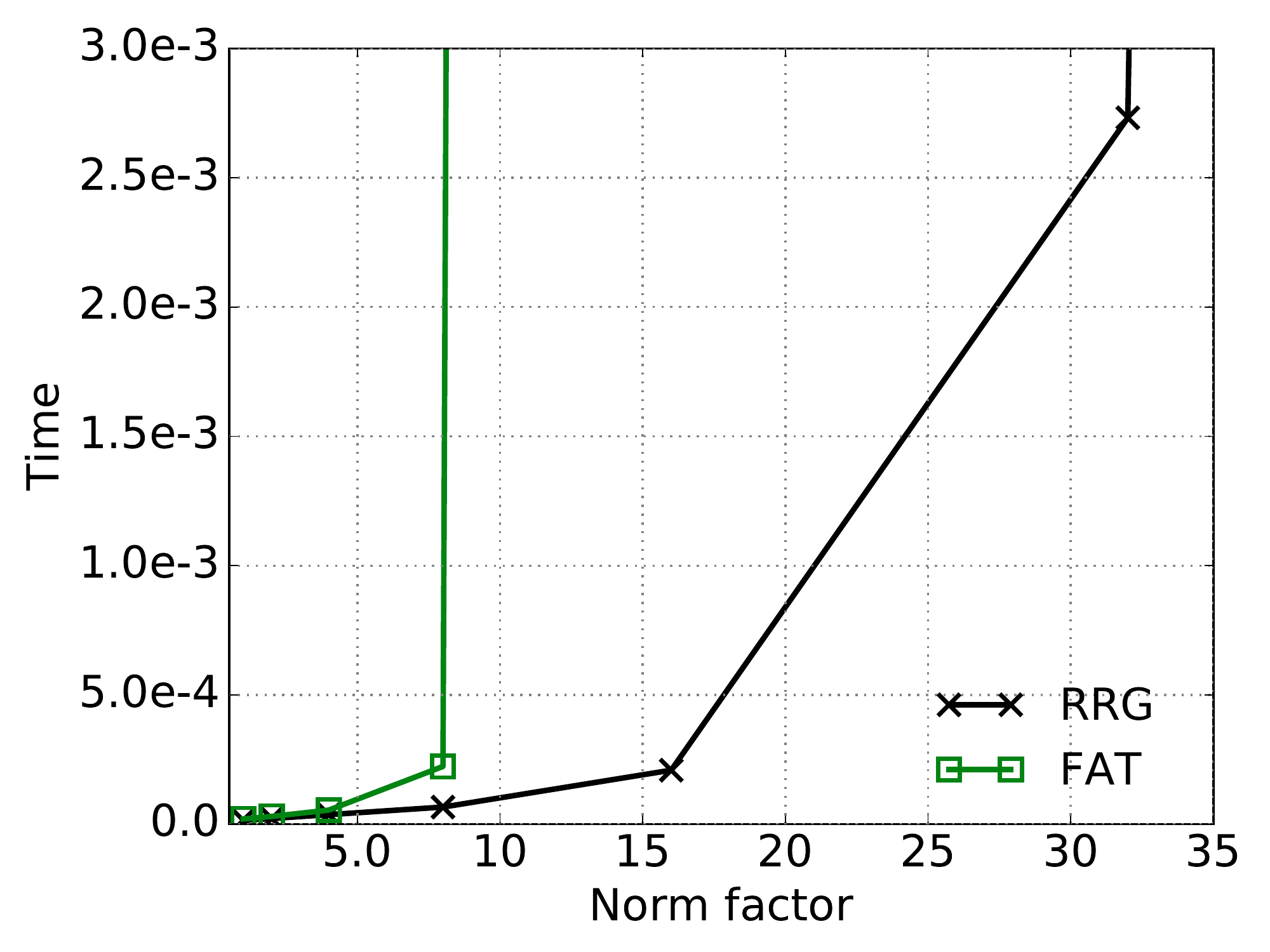}
    \compactcaption{90 \%ile FCT}
    \label{fig:fb 90}
    \end{subfigure}
  \end{minipage} 
  \hskip 0.25cm
  \begin{minipage}[b]{0.3\textwidth}
    \centering
    \begin{subfigure}[b]{\linewidth}
    \includegraphics[width=\textwidth]{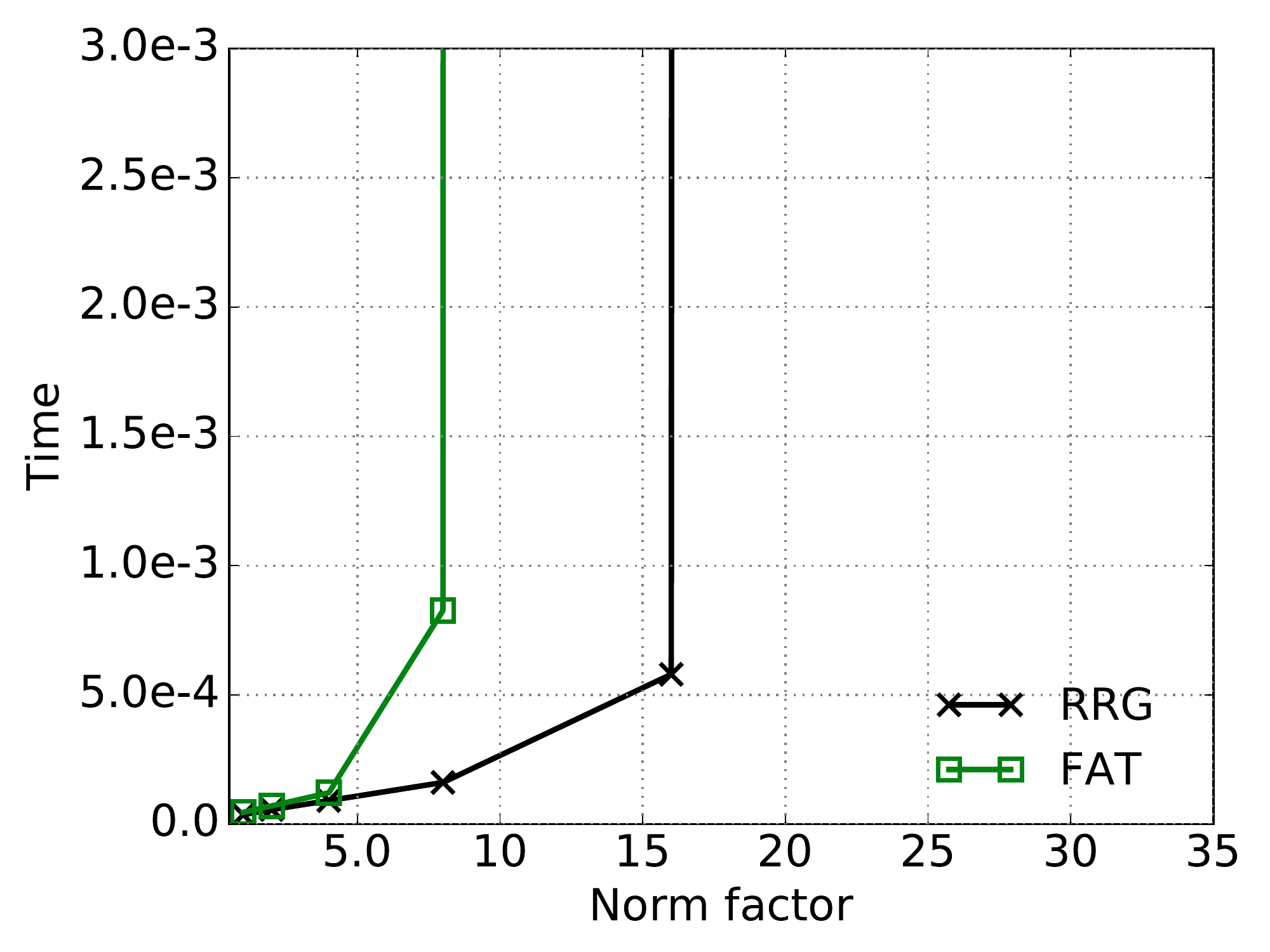}
    \compactcaption{99 \%ile FCT}
    \label{fig:fb 99}
    \end{subfigure}
  \end{minipage}
\caption{Flow completion times at various percentiles for the FB traffic matrix. On the x-axis is the norm factor and hence network load increases along the right of the x-axis.}
\end{figure*}

\subsubsection{Effect of scale}
To examine if the advantages of expanders hold at all scales, we experimented with increasing sizes of leaf spines. For traffic patterns, we chose multiple points in the $C-S$ model as a representative set for the entire space. Figure~\ref{fig:scale experiment} plots the average throughput in the expander (normalized by the average throughput in the leaf spine for that size). It can be seen that the advantages of expanders over leaf spines in the $C-S$ model extend to all scales.

\section{Flow completion time} \label{extreme case analysis}
Traffic in datacenters is often bursty and volatile~\cite{vl209,JupiterRising15}. Incast (many-to-one) and outcast (one-to-many) patterns occur frequently in data processing pipelines. In the following subsection, we discuss multiple experiments in the $C-S$ model, with specific values of $C$ and $S$, to replicate different types of bursty incast and outcast situations. In each case, we use flow completion time as the performance metric.
Every client in $C$ sends a 100 KB flow to every server in $S$. To get maximum burstiness, we start all flows at the same time. 

\subsection{Incasts and Outcasts}
Incasts and outcasts are challenging conditions for the network to deal with. Incasts (and outcasts) could be grouped into two categories~\cite{JupiterRising15}: (1) where the bottleneck is at the host and (2) bottleneck is in the network. Not much can be done by the network for the first kind, when the bottleneck is at the hosts. These type of incasts are mainly many-to-one server level patterns where the bottleneck is due to the traffic fanning in
from the ToRs to a single host~\cite{JupiterRising15}. These constitute majority of the packet drops- 62.8\% according to the study in~\cite{JupiterRising15}. We ran many-to-one (at the server level) incast patterns in ns3 and got roughly the same distribution of flow completion times for both fat trees and the random graph, which is expected since the bottleneck is at the hosts and not the network.

The second type of packet drops are caused mainly due to the ToR fan-in at the aggregation layer (35.9\% of the packet drops~\cite{JupiterRising15}).  These can occur for instance, if many servers are trying to send traffic to the same rack (but possibly different servers in the rack). To reproduce such incasts, we design an experiment involving 40 servers (2 racks in the fat tree) sending traffic to 20 other servers at the same time ($|C| = 40$, $|S| = 20$ in the $C-S$ model). We used the following settings for our experiments in ns3:

\begin{center}
\vskip -0.2 cm
\begin{tabular}{ l |  c  }
  Queue size & 500 packets (droptail) \\ \hline
  Flow size & 100 KB \\ \hline
  Link speed & 1 Gbps \\ \hline
  Link delay & 1 $\mu$sec \\ \hline
  Packet size & 1 KB \\ \hline
  $\text{RTO}_{\text{min}}$ & 5 ms \\ \hline
  Connection Time Out & 50 ms 
\end{tabular}
\end{center}

\begin{figure*}
 \vskip -0.8 cm
\begin{minipage}[b]{0.48\textwidth}
\centering
\includegraphics[width=0.75\textwidth]{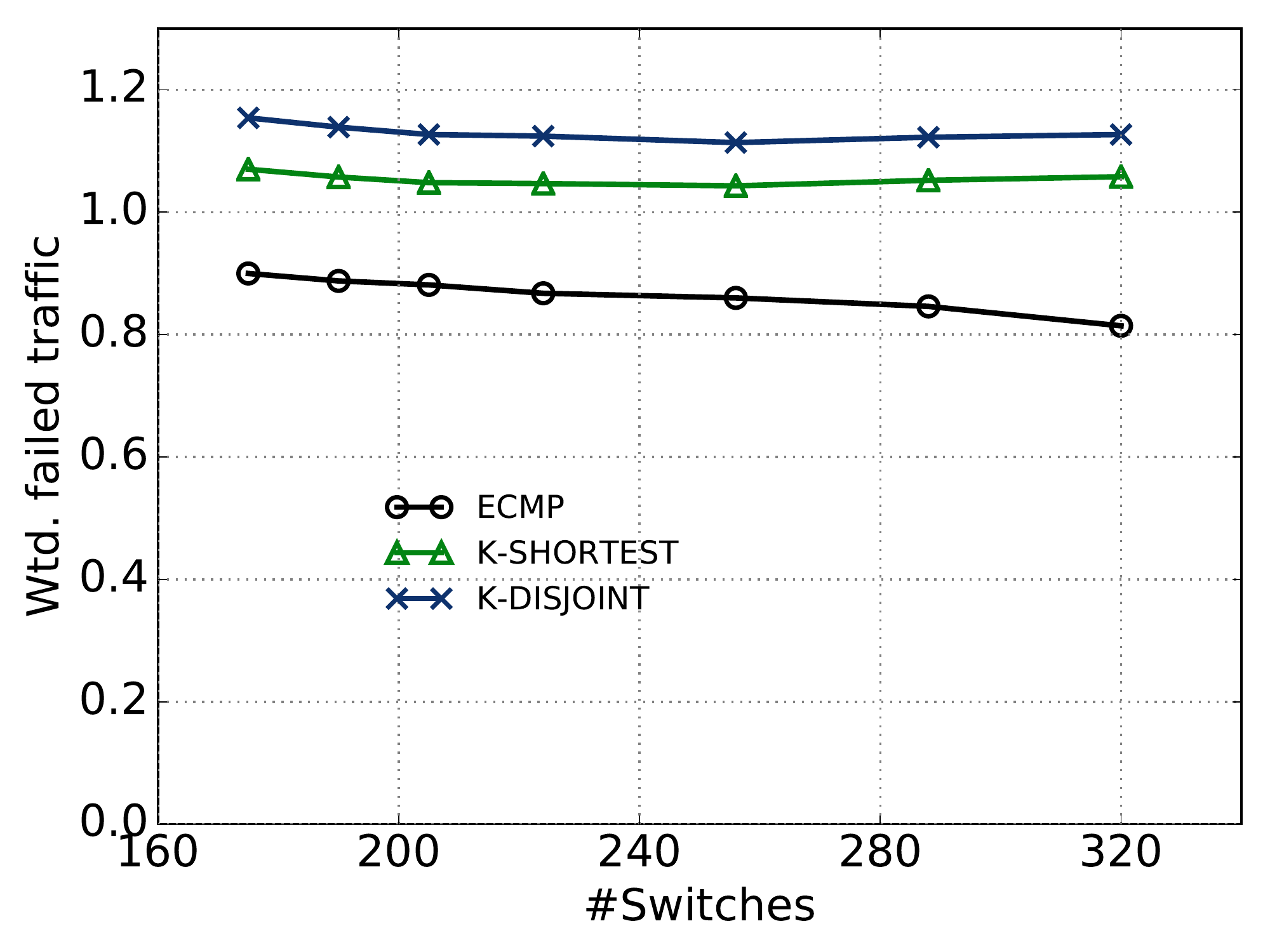}
\compactcaption{Transient traffic loss in a random network due to a single link failure, normalized by traffic loss in a fat tree.}
\label{fig:linkFailure}
\end{minipage}
\hskip 0.6cm 
\begin{minipage}[b]{0.48\textwidth}
\centering
\includegraphics[width=0.75\textwidth]{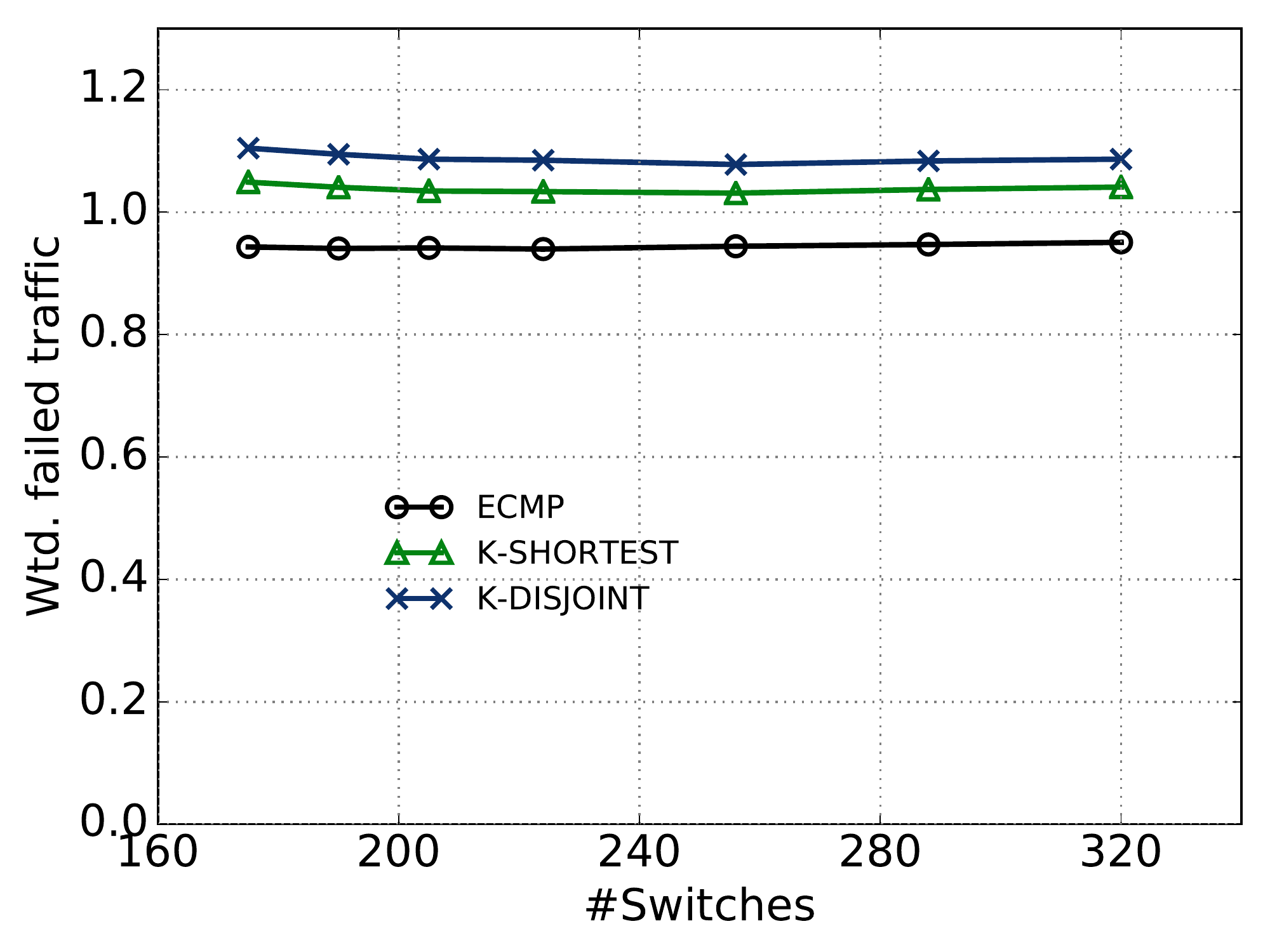}
\compactcaption{Transient traffic loss in a random network due to a single switch failure, normalized by traffic loss in a fat tree.}
\label{fig:switchFailure}
\end{minipage}
\end{figure*}

We used a larger queue size to ensure that the packet loss is not too high to render the network unusable.
 Figure~\ref{fig:incast flow dist} shows the distribution of flow completion times of the $40 \times 20 = 800$ flows involved in the experiment. Clearly, the expander network results in significantly better flow completion times (by more than $2\times$ at the median). 
Figure~\ref{fig:queue buildup rrg} and ~\ref{fig:queue buildup fat} illustrate the queue lengths at the 5 most loaded link ports with time on the x-axis. We verified that all five of the most loaded queues in the fat-tree were at the downlink from the aggregation layer to the ToR. As can be seen in the figure, the traffic jam gets cleared more quickly in the expander graph (roughly twice as quickly except at the busiest link). 

The third kind of packet drops are a result of oversubscription of ToR uplinks (around 1.3\% as per~\cite{JupiterRising15}). Expanders, being flat networks, are better suited to handle such bursts as each rack supports fewer servers. To reproduce such outcast patterns, we ran an experiment with 20 servers (1 rack in the fat tree) sending traffic to 40 other servers ($|C| = 20$, $|S| = 40$ in the $C-S$ model). Figure~\ref{fig:outcast flow dist} shows the distribution of flow completion times for 800 flows. Figures~\ref{fig:queue buildup rrg outcast} and ~\ref{fig:queue buildup fat outcast} show the queue lengths over time for five most loaded links. All five of the most loaded links in the fat tree were ToR uplinks. It can be seen that, in this case too, the queues decline considerably more quickly for the expander network.

\subsection{Datacenter traffic trace}
The $C-S$ model served our purpose for characterizing throughput and latencies for a wide range of cases and identifying cases where a topology does well and where it doesn't. In this section, we describe experiments with real traffic patterns based on packet traces from Facebook's publicly available dataset~\cite{fbTrafficPattern}. 

\subsubsection{Instrumentation} \label{fb tm method}
The Facebook dataset contains packet logs from multiple clusters, we focus on one of the clusters in the datacenter so that we have a complete traffic matrix for that cluster.  We obtain aggregate traffic volumes across racks in that cluster (Figure~\ref{fig:fb rack tm}), disregarding traffic going out of the cluster. Because of the high sampling ratio of 1:30000, we couldn't derive any meaningful conclusions about the flow size distributions. To get the number of racks down to the number of racks in our fat tree topology (50 ToRs), we picked the busiest 50 racks out of 157 racks in the FB cluster to get a rack level traffic matrix (see Figure~\ref{fig:fb rack tm}) for the fat tree. We found that the outgoing traffic volume from a rack has a near uniform distribution (Figure~\ref{fig:fb distr}).
We then split the traffic across a pair of racks equally across all pairs of servers in the two racks. This lead to a server level traffic matrix for the fat tree. To get a server level traffic matrix for the random graph from the server level matrix of the fat tree, we packed the busiest servers into a rack of the random graph, while randomly chosing racks. This process is illustrated in Figure~\ref{fig:fb method}. 
We run this server level traffic matrix in the ns3 simulator~\cite{ns3}, by randomly starting flows across a span of 10 ms. We multiplied each entry in the traffic matrix by a normalizing factor, tuned so that the traffic matrix is neither too small nor too large. We used the following settings for our experiments

\begin{center}
\vskip -0.2 cm
\begin{tabular}{ l |  c  }
  Queue size & 100 packets (droptail) \\ \hline
  Link speed & 1 Gbps \\ \hline
  Link delay & 1 $\mu$sec \\ \hline
  Packet size & 1 KB 
\end{tabular}
\end{center}

\subsubsection{Expanders degrade more gracefully}
By increasing the load on the network, we assess how quickly the network degrades. After a certain load threshold, the network expectedly collapses owing to high packet loss. Things get significantly worse when TCP's control packets start getting dropped. The network is no longer usable at this point. From a design point of view, aside from performance, one also needs to know the network's limits for capacity building for worst case scenarios. 

Figures~\ref{fig:fb 50},~\ref{fig:fb 90} and  ~\ref{fig:fb 99} illustrate flow completion times for the traffic pattern described in Section~\ref{fb tm method}. On the x-axis is the normalizing factor and hence the network load increases along the  x-axis. The results suggest that not only do expanders have a better flow completion time than fat trees, but the network degrades more slowly and gracefully in expanders. In fact, expanders can support more than 2x the network load before degrading beyond being usable.




\section{Traffic Loss on Failures}
Fault recovery in modern data centers~\cite{Walraed-Sullivan2013} is largely based on reconvergence of routing algorithms, resulting in traffic loss till routing tables are recomputed. To alleviate traffic loss, switches react locally to failures by avoiding failed next hops. This is possible only if the routing table has alternate next hops for a given destination. 
For instance, in a fat tree, a source to destination path comprises of an upward path from the source to a core switch and a downward path from the core switch to the destination. Alternate next hops are available only for links in the upward path. If a link on a downward path fails, then a packet traversing that path would get dropped. 
We quantify traffic loss due to a single link failure and a single switch failure till routing tables are recomputed. We assume local convergence, that is, switches are capable of avoiding a failed link if alternate next hops are available for that destination.

\begin{figure*}
 \vskip -0.4 cm
    \begin{minipage}[b]{0.3\textwidth}
    \centering
    \begin{subfigure}[b]{\linewidth}
    \hskip 0 cm
    \includegraphics[width=0.9\textwidth,height=0.7\textwidth]{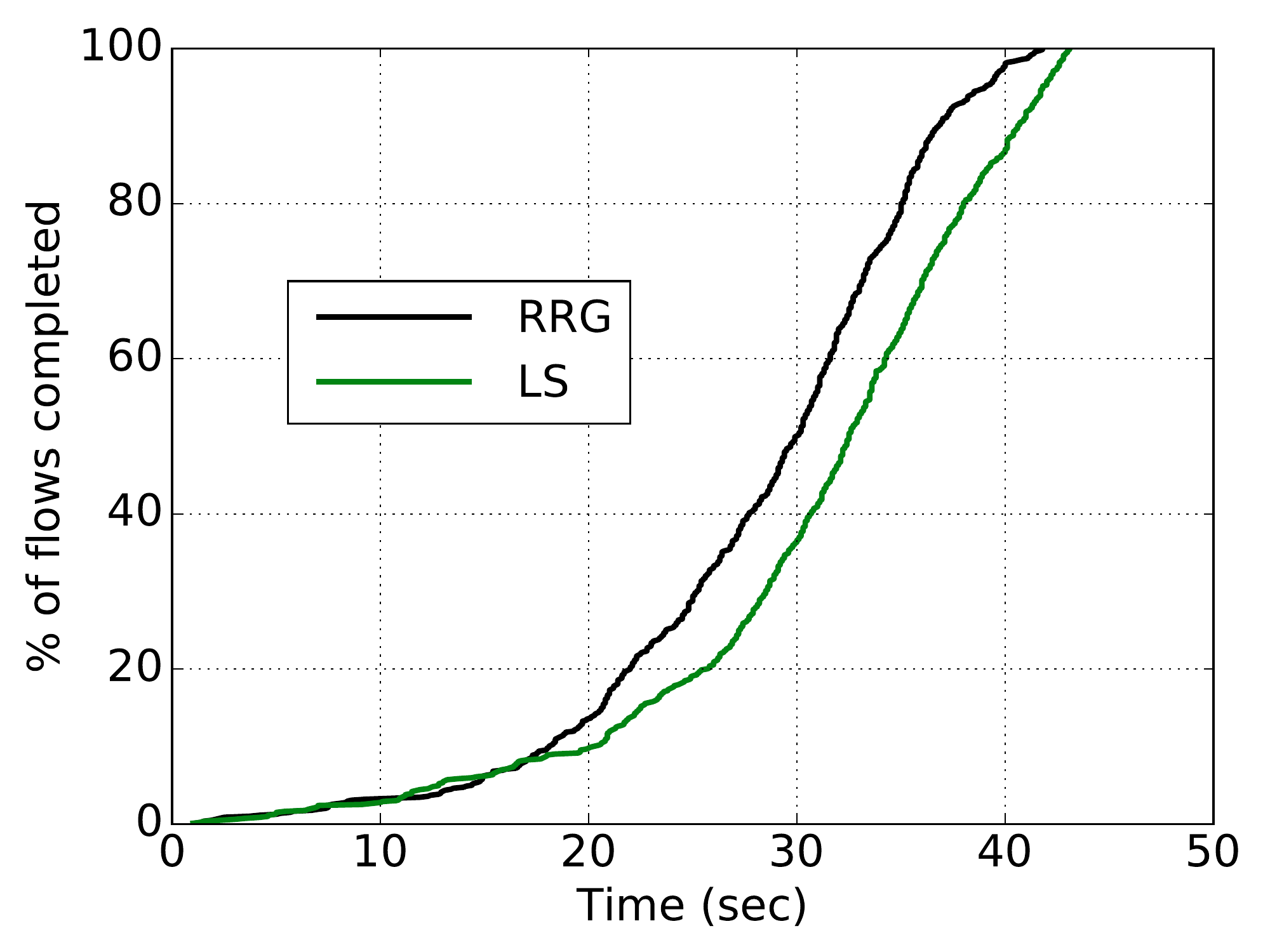}
    \compactcaption{Distribution of flow completion times for a 100 MB shuffle operation on the testbed for a leaf spine (3 spines, 7 leaf ToRs, 28 servers) and an equivalent random graph (RRG). }
    \label{fig:hardware_shuffle_flowdist}
  \end{subfigure}
  \end{minipage}
  \hskip 0.35 cm
  \begin{minipage}[b]{0.3\textwidth}
    \centering
    \begin{subfigure}[b]{\linewidth}
  \includegraphics[width=1.15\textwidth]{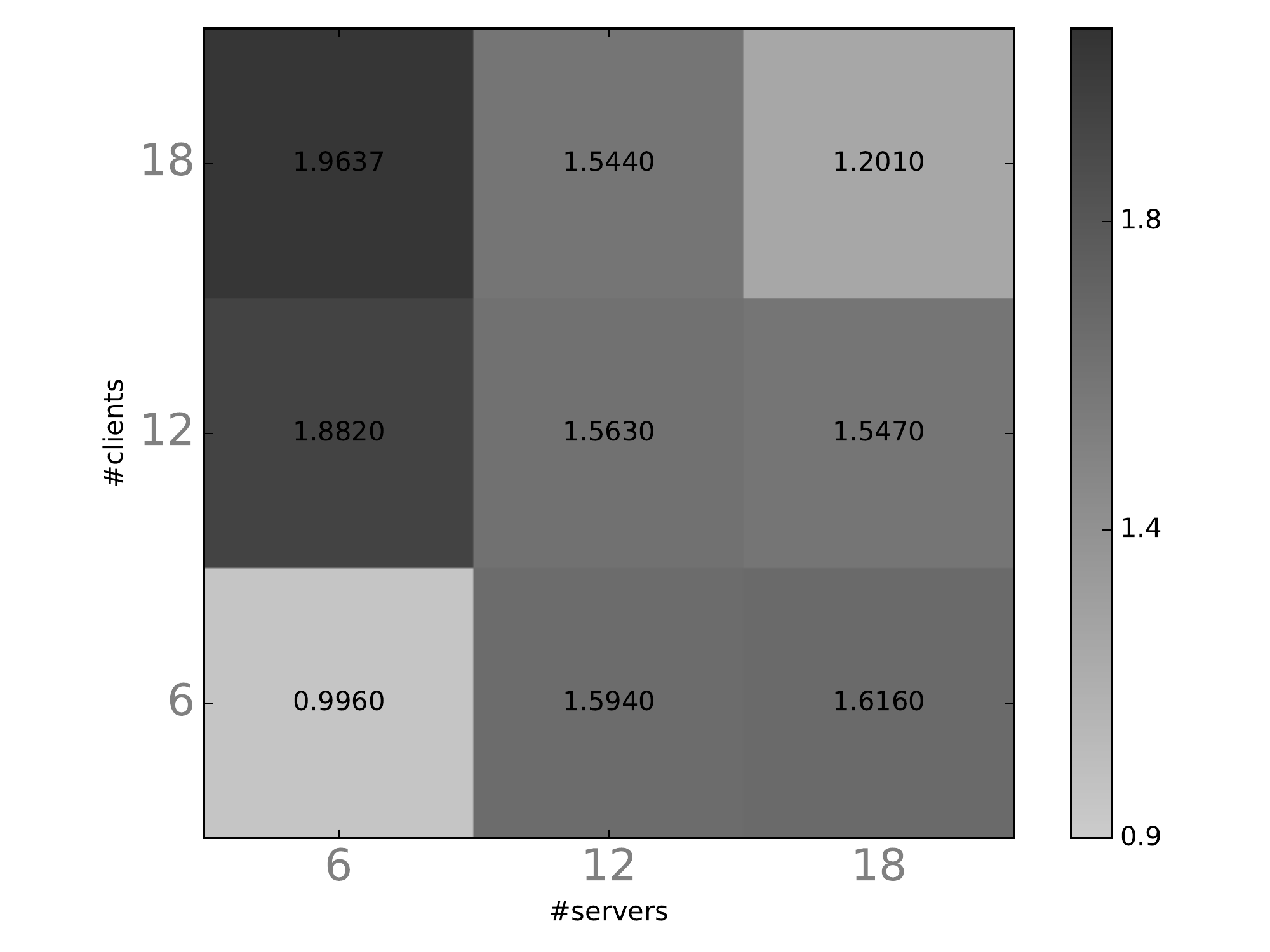}
  \compactcaption{Simulated on htsim}
    \label{fig:cs_simulation}
  \end{subfigure}
  \end{minipage}
  \hskip 0.45 cm
  \begin{minipage}[b]{0.3\textwidth}
    \centering
    \begin{subfigure}[b]{\linewidth}
    \includegraphics[width=1.15\textwidth]{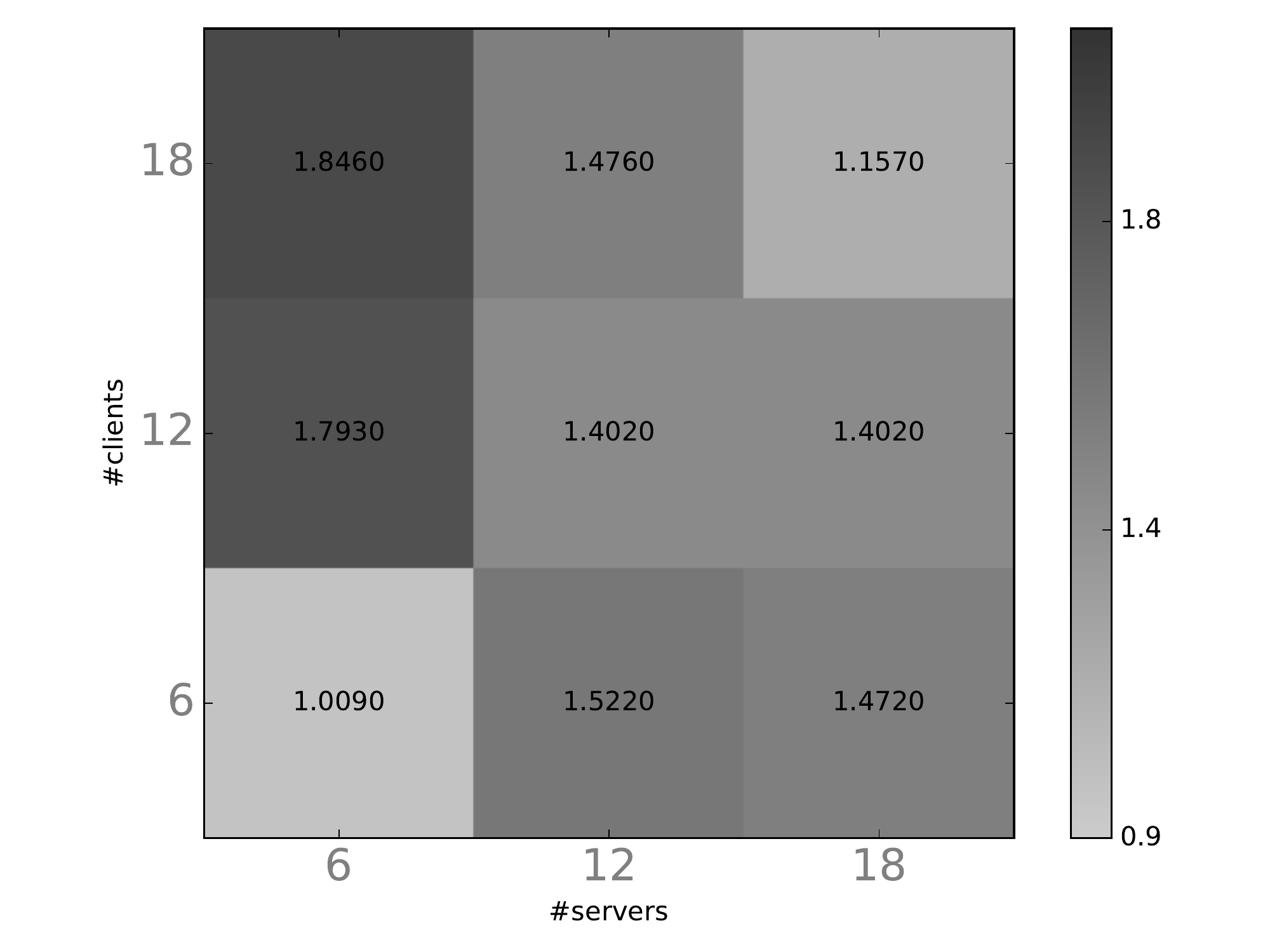}
    \compactcaption{Hardware testbed}
    \label{fig:cs_hardware}
  \end{subfigure}
  \end{minipage} 
  \label{fig:hardware_experiments}
  \caption{Hardware testbed experiments. Figure (a) illustrates the flow completion times for a shuffle operation. Figures (b) and (c) illustrate experiments in the htsim simulator and the hardware testbed respectively, representing coarse-grained C-S model runs for LeafSpine(x=6, y=2) consisting of 48 servers and an equivalent random graph. Each point is the ratio of the average throughput in the random graph divided by the average throughput in leaf spine.}
\end{figure*}

\subsection{Model}
We model traffic loss as follows. Assume one unit of traffic flow between a random source and destination pair. Further, assume that the path for the flow, is chosen according to an ideal hash function at each switch, that hashes based on the tuple identifying the flow. If traffic arrives at a switch that has no active next hops, the entire flow gets dropped. 
Let F denote the event of a single link (switch) failure and L denote the event of traffic getting dropped between a random source and destination pair. Then, traffic loss due to a single link (switch) failure can be given by $P[L|F] \times P[F]$. We have,
\begin{align*}
P[F] \approx \lambda \times \text{Number of links (switches)}
\end{align*}
$\lambda$ is a constant that denotes the probability of a link (switch) failing in a given interval of time. $\lambda$ is also related to the \textit{Mean Time To Failure} (MTTF) of a single link (switch). 

We compute $P[L|F]$ empirically for shortest path routing, k-shortest and k-disjoint path routing. Traffic loss for k-shortest and k-disjoint path routing depends on the implementation, we assume complete source routing. Local convergence would imply that if there's a path for which the first link has failed, then the source switch would avoid that route and choose from  alternate paths. If there are no alternate paths, all packets in the flow would get dropped till the routing algorithm reconverges.

\subsection{Single Link and Switch Failure}
Traffic loss is determined by two factors - probability of hitting a failed link and the ability to locally reroute to alternate next hops. The probability of hitting a failed link is infact, directly proportional to the average path length. The route lengths are largest in k-disjoint followed by k-shortest followed by shortest path routing which employs paths of optimal lengths. We see that the same is also reflected in Figures~\ref{fig:linkFailure} and ~\ref{fig:switchFailure}. We emphasize that marginal differences in traffic loss for different routing schemes might not be of significant importance, since any traffic loss is short-lived lasting only till the routing algorithm reconverges. Both figures~\ref{fig:linkFailure} and ~\ref{fig:switchFailure} convey the same underlying message:  all schemes result in acceptable traffic loss in comparison to shortest path routing in a fat tree.



\section{Hardware Experiments}
In this section, we present experiments on a physical testbed 
consisting of 13 OpenFlow-enabled Pica8 switches, connecting to several VMs over a 1Gbps link. We emulate a leaf spine network with 3 spines, 7 ToR leaves and 28 servers, and a random graph topology with the same equipment. 
We study a shuffle operation, a common primitive among data processing jobs (e.g. to hash-partition keys in a map-reduce job). 
The shuffle operation involves an iperf~\cite{iperf2012} data transfer across every pair of nodes in the system. 
We used the following settings for the experiment:
\begin{center}
\begin{tabular}{ l |  c  }
  Transport & TCP \\ \hline
  TCP Window & 416 KB \\ \hline
  Flow size & 100 MB  \\ \hline
  Link capacity & 1 Gbps \\ \hline
  Routing & Shortest path \\ \hline
  Load balancing & ECMP 
\end{tabular}
\end{center}

Figure~\ref{fig:hardware_shuffle_flowdist} illustrates the distribution of flow completion times for the shuffle operation. The all-to-all heavy shuffle keeps the entire network busy and hence there is little room for spreading out the traffic. Hence, the expander doesn't perform significantly better than the leaf spine. Note that the expander graph is still a little more efficient because it uses shorter paths.

Figures~\ref{fig:cs_hardware} and ~\ref{fig:cs_simulation} illustrate coarse grained $C-S$ model runs, comparing the simulator and the hardware results. As can be seen, the benefits of expander graphs are notable even at such small scale. The baseline topology used for these experiments was a leaf spine network with 2 spines and 8 leaf switches supporting a total of 48 servers. 


\subsection{Implementing k-disjoint Routing}\label{kshortimpl}


\begin{figure}
\centering
\includegraphics[width=0.16\textwidth]{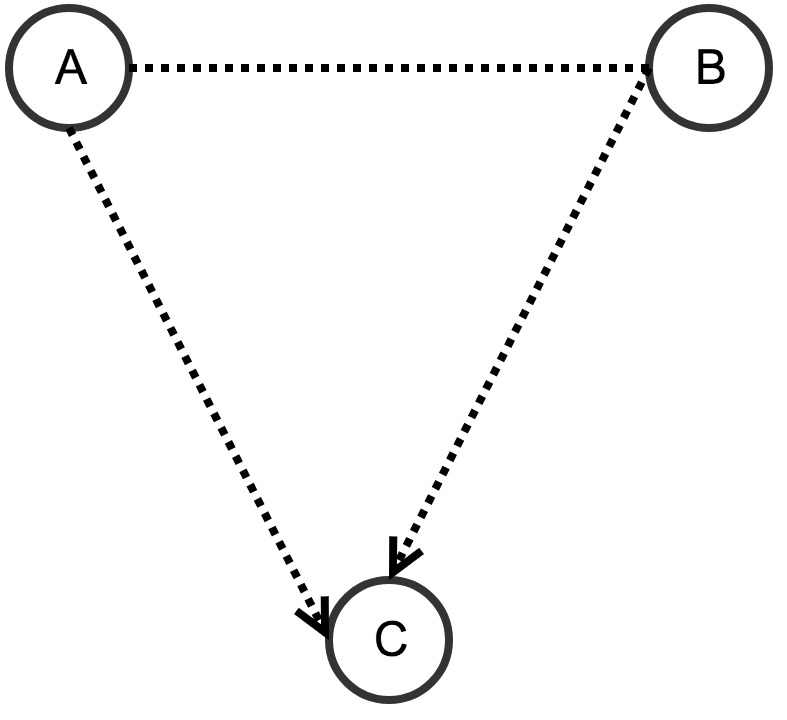}
\compactcaption{Consider k-disjoint paths from A to C, where k=2. One of the paths go through B. Similarily, one of the k-disjoint paths from B to C go through A. Hence, using only destination based next-hop selection at the switches would lead to a routing loop.}
\label{fig:kdisjointDestination}
\end{figure}

Implementing k-disjoint and k-shortest paths routing is non-trivial with commodity switches. Many switches support only destination based look-ups, such as longest prefix matches on destination IP addresses. It can be easily seen that a network fabric with only destination based look-ups can not support k-disjoint path or k-shortest path routing trivially (see Figure~\ref{fig:kdisjointDestination}). If matching on arbitrary bitmasks were allowed, then source routing could be implemented efficiently in openflow switches as shown in~\cite{jyothi15} and ~\cite{guo2010}. One could easily implement k-shortest path routing using source routing. However, most commodity openflow switches do not support matching on arbitrary bitmasks despite it being in the OpenFlow specification~\cite{openflow_v1.4}. Alternatively, one could stack MPLS labels in the packet header to implement source routing. This has two problems a) it might blow up the packet header size since each MPLS label is 4 Bytes b) typically, there is a limit to the number of MPLS labels one can stack up in the packet header.

We take another approach that uses a simple observation. Call a path $P$ from source $s$ to destination $t$ \textbf{expressible} if there is an intermediate hop $u$ in $P$ such that $P_{s\rightarrow u}$ and $P_{u\rightarrow t}$ are shortest paths, where $P_{x\rightarrow y}$ denotes the portion of $P$ from $x$ to $y$. If the network uses standard ECMP routing, then the routing entry for $P$ can be expressed by attaching a single MPLS label for $u$ at source $s$. Note that pushing a MPLS label for $u$ and using ECMP between $s$ to $u$ and $u$ to $t$ is not the same as $P$ since there could be multiple shortest paths from $s$ to $u$ and $u$ to $t$. However, this only increases path diversity without increasing the path length and hence will only help. We found out that most of the k-disjoint paths in k-disjoint path routing could be expressed as a concatenation of two shortest paths. Figure~\ref{fig:expressibility} illustrates that the number of k-disjoint paths that are not expressible is less than 0.01\%. Expressing routing paths by bouncing the packets off an intermediate hop also makes it possible to implement it using segment routing~\cite{segment_routing} which is increasingly getting adopted~\cite{segment_routing_cisco,segment_routing_juniper}.

\begin{figure}
\centering
\includegraphics[width=0.37\textwidth]{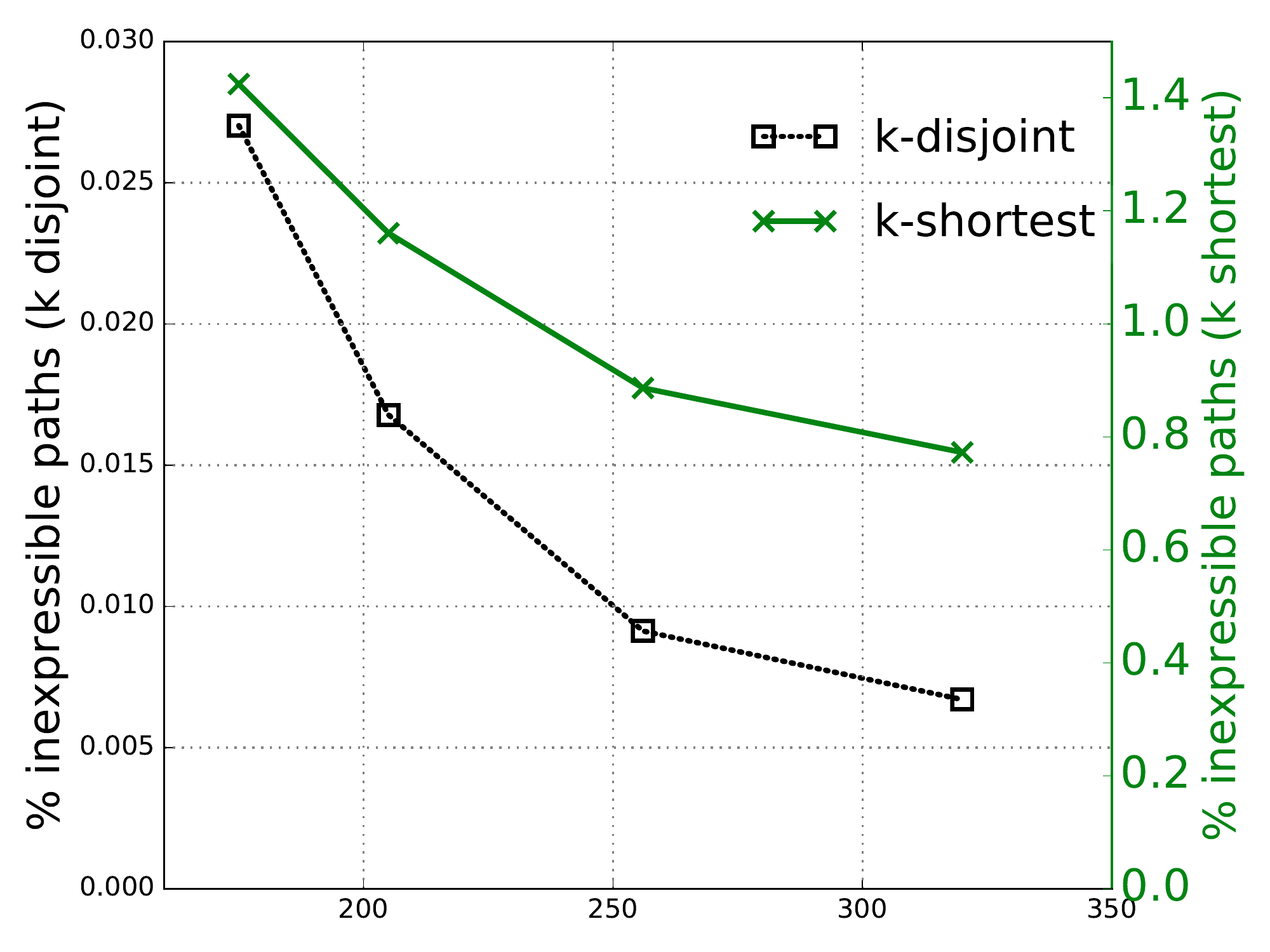}
\compactcaption{Nearly all k-disjoint paths are expressible as a concatenation of two shortest paths. The y-axis represents the percentage of paths that are not expressible. The y-axis on the right represents the same quantity for k-shortest paths.}
\label{fig:expressibility}
\end{figure}

\section{Optimizing long distance links} \label{cross links}
Managing wiring complexity of an expander network poses a challenge for deployment, especially on large scale. While random graphs with more localized short distance links have been proposed as a solution to better manage the complexity~\cite{singla14}, such networks would lose out on performance. In appendix~\ref{cabling and datacenter expansion}, we show that the expansion ratio of an expander (which can be thought of as a metric to judge the expander quality) decreases at least linearly with decreasing fraction of long distance links. We take another approach instead, one that is applicable for small to medium sized networks and does not make any compromise on the quality of the expander graph. We used the publicly available Metis graph partitioning library~\cite{metis} to partition the network into $k$ multiple equal-sized clusters (in terms of the number of racks), minimizing the number of cross-cluster long distance links. Note that this problem is NP-hard in general, even for $k=2$ clusters, so the Metis software uses  heuristics that are known to work well in practice. Figure~\ref{fig:cross links} shows the number of cross-cluster links after graph partitioning. Without partitioning, the expected number of cross-cluster links in a random graph, partitioned into $k$ clusters, would be $(k-1)/k$. Thus, for a random network with $320$ switches, smartly partitioning the network into $k=5$ partitions brings down the number of long distance links from $80\%$ to approximately $50\%$ (see Figure~\ref{fig:cross links}). The partitioning trick doesn't however work for hyperscale sized networks. As can be seen in Figure~\ref{fig:cross links}, as the network size increases, the fraction of cross cluster links after partitioning increases.

\begin{figure}
\centering
\includegraphics[width=0.43\textwidth]{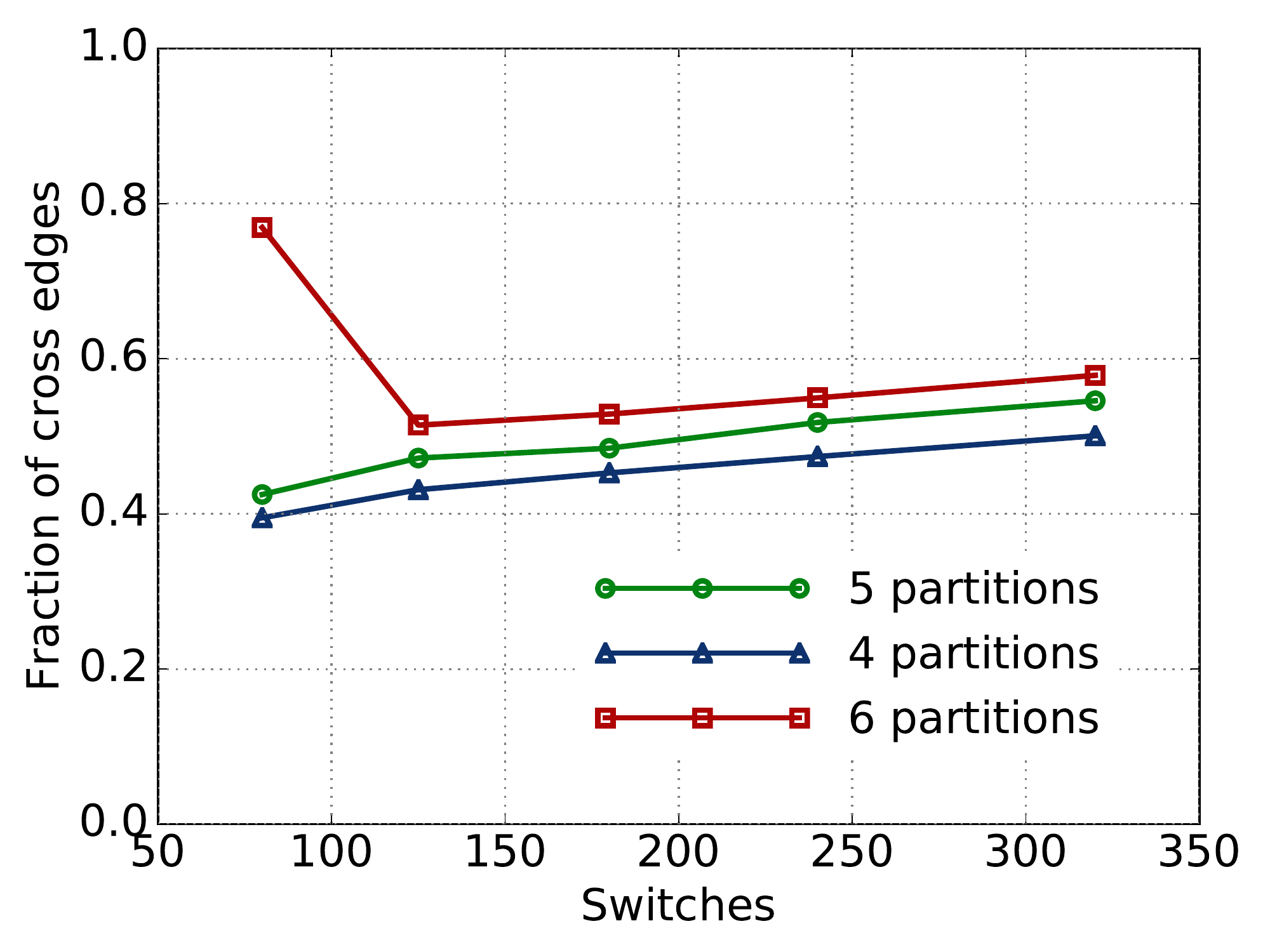}
\compactcaption{Fraction of cross cluster long-distance links after partitioning the random graph into multiple clusters using a graph partitioning algorithm.}
\label{fig:cross links}
\end{figure}

\section{Related work}
\subsection{Tree-based designs}
Apart from the standard 3-tier fat tree~\cite{AlFares08}, a lot of research has been done to improve tree-based designs.
LEGUP~\cite{legup} proposed a design to build heterogeneous Clos networks for flexible expansion. F10~\cite{F10}, although inherently similar to Fat trees, reacts quickly to failures and preserves performance even in the face of multiple failed nodes by co-designing routing algorithms, network and failure handling mechanism. BCube~\cite{bcube} recursively builds a tree-based server-centric architecture and is applicable mainly to container based datacenters that require high failure resilience.

\subsection{Non-tree based static networks}
Slimfly~\cite{besta14} optimizes for low network diameters and is also based on expander graphs. The proposal relies heavily on non-oblivious routing schemes that requires information about queues either locally or globally~\cite{besta14}. Further, it loses out on performance for having a smaller expansion ratio than Jellyfish or Xpander.
Dragonfly~\cite{dragonfly} uses groups of high radix routers as virtual routers, effectively increasing the network radix to achieve low diameter. However, reliance on randomized routing based on Valiant's algorithm to route packets across groups poses challenges for adoption. 
Designs based on structured sub-graphs and random links between the structured sub-graphs~\cite{smallWorldDatacenters}, inspired from the small-world distribution, were proposed to achieve high failure resilience and bandwidth. These designs have a lot in common with the random graph topology. The Flat-tree~\cite{Flat-tree} architectures explores a convertible design between Clos topologies and random graphs to achieve the best of both worlds.

\subsection{Dynamic Networks}
Another research thrust has been towards dynamic networks, where link connections are configured dynamically based on the traffic load using free space optics. The argument for dynamic networks is the observation that traffic patterns in datacenters are highly skewed with only a few racks sending and receiving most of the traffic at any given point in time~\cite{ghobadi16}. Proposals connecting racks with ceiling mirrors~\cite{hamedazimi14,zhou12}, beam redirecting ceiling balls~\cite{ghobadi16} that provide a large fan-out and using wireless links as an alternative to wired links~\cite{OSA,helios,MegaSwitch,cthrough}.  The advantage of such designs is that by varying the tilt of the reflecting surfaces, the incoming traffic beam can be redirected to a large number of racks. However, slow switching time of links (around 10-100 ms) poses a challenge in adoption of dynamic networks. Mordia~\cite{mordia} uses circuit scheduling to proactively assign circuit bandwidths and achieves a much lower switching time of around 11.5$\mu$s.

\subsection{Traffic Engineering}
Several works in the past have aimed at improving load balance in Clos topologies to minimize flow completion times. Some of these ideas can be extended to expander networks. CONGA~\cite{CONGA} uses flowlet switching~\cite{flowletSwitching1}, choosing paths based on congestion levels on the paths. PDQ~\cite{PDQ} uses preemptive flow scheduling to simulate a number of scheduling schemes like shortest job first and earliest deadline first. Drill~\cite{Drill} makes fine grained decisions at the packet level based on load on the next hop. Presto~\cite{Presto} evenly distributes fixed sized flow-lets called flow-cells across all paths to the destination. Specific traffic engineering and load balancing schemes targeted towards expander networks would be an interesting avenue of research which we leave for future work.

\section{Conclusion}

We discovered that expander datacenters are extremely effective even with just traditional protocols, offering around 3- 4$\times$ and 1.5- 2$\times$ more bandwidth on average compared to fat trees (large scale) and leaf spine (small scale) networks respectively, for a wide range of cases. The advantages of expanders are even more pronounced with k-disjoint path routing, that we showed, can be implemented with segment routing. We examined extreme traffic situations like incasts and outcasts and showed that expanders are more resilient to such cases and degrade more gracefully as load increases. Furthermore, we analyzed several other metrics of interests: traffic loss on failures, time taken to complete a shuffle operation, queue occupancy and fraction of long distance links. Our results show that non-oblivious schemes or complex traffic engineering schemes are not required for building an expander datacenter.
Our experimental results, based on extensive simulation and hardware emulation, unanimously identify significant advantages of using expander-based datacenters over tree-based networks at both small scale and large scale with protocols that are realizable with today's hardware.

{\footnotesize \bibliographystyle{acm}
\bibliography{bibliography}}

\begin{thebibliography}{10}

\bibitem{segment_routing_cisco}
Cisco segment routing.
\newblock
  \url{https://www.cisco.com/c/en/us/solutions/service-provider/cloud-scale-networking-solutions/segment-routing.html}.

\bibitem{AristaBestPractices}
Cloud networking scale out - arista.
\newblock
  \url{https://www.arista.com/assets/data/pdf/Whitepapers/Cloud_Networking\_\_Scaling\_Out\_Data_Center\_Networks.pdf}.

\bibitem{fbTrafficPattern}
Data sharing on traffic pattern inside facebook’s datacenter network.
\newblock \url{https://research.fb.com/data-sharing-on-traffic-pattern-inside-\
  facebooks-datacenter-network/}.

\bibitem{kowalski2013expander}
Expander graphs.
\newblock \url{https://people.math.ethz.ch/~kowalski/expander-graphs.pdf}.

\bibitem{fbDcFabric}
Introducing data center fabric, the next-generation facebook data center
  network.
\newblock \url{https://code.facebook.com/posts/360346274145943}.

\bibitem{segment_routing_juniper}
Juniper segment routing.
\newblock
  \url{https://www.juniper.net/documentation/en_US/northstar3.2.0/topics/concept/northstar-spring.html}.

\bibitem{metis}
Metis software.
\newblock \url{http://glaros.dtc.umn.edu/gkhome/metis/metis/overview}.

\bibitem{ns3}
Ns3.
\newblock \url{https://www.nsnam.org/}.

\bibitem{segment_routing}
Segment routing architecture.
\newblock \url{https://tools.ietf.org/html/rfc8402}.

\bibitem{AlFares08}
{\sc Al-Fares, M., Loukissas, A., and Vahdat, A.}
\newblock A scalable, commodity data center network architecture.
\newblock In {\em Proceedings of the ACM SIGCOMM 2008 Conference on Data
  Communication\/} (New York, NY, USA, 2008), SIGCOMM '08, ACM, pp.~63--74.

\bibitem{CONGA}
{\sc Alizadeh, M., Edsall, T., Dharmapurikar, S., Vaidyanathan, R., Chu, K.,
  Fingerhut, A., Lam, V.~T., Matus, F., Pan, R., Yadav, N., and Varghese, G.}
\newblock Conga: Distributed congestion-aware load balancing for datacenters.
\newblock {\em SIGCOMM Comput. Commun. Rev. 44}, 4 (Aug. 2014), 503--514.

\bibitem{benson2010}
{\sc Benson, T., Akella, A., and Maltz, D.~A.}
\newblock Network traffic characteristics of data centers in the wild.
\newblock In {\em Proceedings of the 10th ACM SIGCOMM Conference on Internet
  Measurement\/} (New York, NY, USA, 2010), IMC '10, ACM, pp.~267--280.

\bibitem{besta14}
{\sc Besta, M., and Hoefler, T.}
\newblock Slim fly: A cost effective low-diameter network topology.
\newblock In {\em Proceedings of the International Conference for High
  Performance Computing, Networking, Storage and Analysis\/} (Piscataway, NJ,
  USA, 2014), SC '14, IEEE Press, pp.~348--359.

\bibitem{MegaSwitch}
{\sc Chen, L., Chen, K., Zhu, Z., Yu, M., Porter, G., Qiao, C., and Zhong, S.}
\newblock Enabling wide-spread communications on optical fabric with
  megaswitch.
\newblock In {\em 14th {USENIX} Symposium on Networked Systems Design and
  Implementation ({NSDI} 17)\/} (Boston, MA, 2017), {USENIX} Association,
  pp.~577--593.

\bibitem{legup}
{\sc Curtis, A.~R., Keshav, S., and Lopez-Ortiz, A.}
\newblock Legup: Using heterogeneity to reduce the cost of data center network
  upgrades.
\newblock In {\em Proceedings of the 6th International COnference\/} (New York,
  NY, USA, 2010), Co-NEXT '10, ACM, pp.~14:1--14:12.

\bibitem{delimitrou2012}
{\sc Delimitrou, C., Sankar, S., Kansal, A., and Kozyrakis, C.}
\newblock Echo: Recreating network traffic maps for datacenters with tens of
  thousands of servers.
\newblock In {\em Proceedings of the 2012 IEEE International Symposium on
  Workload Characterization (IISWC)\/} (Washington, DC, USA, 2012), IISWC '12,
  IEEE Computer Society, pp.~14--24.

\bibitem{hoseModel}
{\sc Duffield, N.~G., Goyal, P., Greenberg, A., Mishra, P., Ramakrishnan,
  K.~K., and van~der Merive, J.~E.}
\newblock A flexible model for resource management in virtual private networks.
\newblock In {\em Proceedings of the Conference on Applications, Technologies,
  Architectures, and Protocols for Computer Communication\/} (New York, NY,
  USA, 1999), SIGCOMM '99, ACM, pp.~95--108.

\bibitem{birkhoffVonNeumann}
{\sc Dufossé, F., and Uçar, B.}
\newblock Notes on birkhoff–von neumann decomposition of doubly stochastic
  matrices.
\newblock {\em Linear Algebra and its Applications 497\/} (2016), 108 -- 115.

\bibitem{ellis2011expansion}
{\sc Ellis, D.}
\newblock The expansion of random regular graphs.

\bibitem{helios}
{\sc Farrington, N., Porter, G., Radhakrishnan, S., Bazzaz, H.~H., Subramanya,
  V., Fainman, Y., Papen, G., and Vahdat, A.}
\newblock Helios: A hybrid electrical/optical switch architecture for modular
  data centers.
\newblock In {\em Proceedings of the ACM SIGCOMM 2010 Conference\/} (New York,
  NY, USA, 2010), SIGCOMM '10, ACM, pp.~339--350.

\bibitem{ghobadi16}
{\sc Ghobadi, M., Mahajan, R., Phanishayee, A., Devanur, N., Kulkarni, J.,
  Ranade, G., Blanche, P.-A., Rastegarfar, H., Glick, M., and Kilper, D.}
\newblock Projector: Agile reconfigurable data center interconnect.
\newblock In {\em Proceedings of the 2016 ACM SIGCOMM Conference\/} (New York,
  NY, USA, 2016), SIGCOMM '16, ACM, pp.~216--229.

\bibitem{Drill}
{\sc Ghorbani, S., Yang, Z., Godfrey, P.~B., Ganjali, Y., and Firoozshahian,
  A.}
\newblock Drill: Micro load balancing for low-latency data center networks.
\newblock In {\em Proceedings of the Conference of the ACM Special Interest
  Group on Data Communication\/} (New York, NY, USA, 2017), SIGCOMM '17, ACM,
  pp.~225--238.

\bibitem{vl209}
{\sc Greenberg, A., Hamilton, J.~R., Jain, N., Kandula, S., Kim, C., Lahiri,
  P., Maltz, D.~A., Patel, P., and Sengupta, S.}
\newblock Vl2: A scalable and flexible data center network.
\newblock In {\em Proceedings of the ACM SIGCOMM 2009 Conference on Data
  Communication\/} (New York, NY, USA, 2009), SIGCOMM '09, ACM, pp.~51--62.

\bibitem{bcube}
{\sc Guo, C., Lu, G., Li, D., Wu, H., Zhang, X., Shi, Y., Tian, C., Zhang, Y.,
  and Lu, S.}
\newblock Bcube: A high performance, server-centric network architecture for
  modular data centers.
\newblock In {\em Proceedings of the ACM SIGCOMM 2009 Conference on Data
  Communication\/} (New York, NY, USA, 2009), SIGCOMM '09, ACM, pp.~63--74.

\bibitem{guo2010}
{\sc Guo, C., Lu, G., Wang, H.~J., Yang, S., Kong, C., Sun, P., Wu, W., and
  Zhang, Y.}
\newblock Secondnet: A data center network virtualization architecture with
  bandwidth guarantees.
\newblock In {\em Proceedings of the 6th International COnference\/} (New York,
  NY, USA, 2010), Co-NEXT '10, ACM, pp.~15:1--15:12.

\bibitem{guo2008}
{\sc Guo, C., Wu, H., Tan, K., Shi, L., Zhang, Y., and Lu, S.}
\newblock Dcell: A scalable and fault-tolerant network structure for data
  centers.
\newblock In {\em Proceedings of the ACM SIGCOMM 2008 Conference on Data
  Communication\/} (New York, NY, USA, 2008), SIGCOMM '08, ACM, pp.~75--86.

\bibitem{hamedazimi14}
{\sc Hamedazimi, N., Qazi, Z., Gupta, H., Sekar, V., Das, S.~R., Longtin,
  J.~P., Shah, H., and Tanwer, A.}
\newblock Firefly: A reconfigurable wireless data center fabric using
  free-space optics.
\newblock In {\em Proceedings of the 2014 ACM Conference on SIGCOMM\/} (New
  York, NY, USA, 2014), SIGCOMM '14, ACM, pp.~319--330.

\bibitem{Presto}
{\sc He, K., Rozner, E., Agarwal, K., Felter, W., Carter, J., and Akella, A.}
\newblock Presto: Edge-based load balancing for fast datacenter networks.
\newblock {\em SIGCOMM Comput. Commun. Rev. 45}, 4 (Aug. 2015), 465--478.

\bibitem{PDQ}
{\sc Hong, C.-Y., Caesar, M., and Godfrey, P.~B.}
\newblock Finishing flows quickly with preemptive scheduling.
\newblock In {\em Proceedings of the ACM SIGCOMM 2012 Conference on
  Applications, Technologies, Architectures, and Protocols for Computer
  Communication\/} (New York, NY, USA, 2012), SIGCOMM '12, ACM, pp.~127--138.

\bibitem{iperf2012}
{\sc Iperf, T.}
\newblock Udp bandwidth measurement tool, 2012.

\bibitem{jyothi15}
{\sc Jyothi, S.~A., Dong, M., and Godfrey, P.~B.}
\newblock Towards a flexible data center fabric with source routing.
\newblock In {\em Proceedings of the 1st ACM SIGCOMM Symposium on Software
  Defined Networking Research\/} (New York, NY, USA, 2015), SOSR '15, ACM,
  pp.~10:1--10:8.

\bibitem{jyothi2016}
{\sc Jyothi, S.~A., Singla, A., Godfrey, P.~B., and Kolla, A.}
\newblock Measuring and understanding throughput of network topologies.
\newblock In {\em Proceedings of the International Conference for High
  Performance Computing, Networking, Storage and Analysis\/} (Piscataway, NJ,
  USA, 2016), SC '16, IEEE Press, pp.~65:1--65:12.

\bibitem{flowletSwitching2}
{\sc Kandula, S., Katabi, D., Sinha, S., and Berger, A.}
\newblock Dynamic load balancing without packet reordering.
\newblock {\em SIGCOMM Comput. Commun. Rev. 37}, 2 (Mar. 2007), 51--62.

\bibitem{kassing17}
{\sc Kassing, S., Valadarsky, A., Shahaf, G., Schapira, M., and Singla, A.}
\newblock Beyond fat-trees without antennae, mirrors, and disco-balls.
\newblock In {\em Proceedings of the Conference of the ACM Special Interest
  Group on Data Communication\/} (New York, NY, USA, 2017), SIGCOMM '17, ACM,
  pp.~281--294.

\bibitem{dragonfly}
{\sc Kim, J., Dally, W.~J., Scott, S., and Abts, D.}
\newblock Technology-driven, highly-scalable dragonfly topology.
\newblock In {\em Proceedings of the 35th International Symposium on Computer
  Architecture\/} (Washington, DC USA, 2008), pp.~77--88.

\bibitem{F10}
{\sc Liu, V., Halperin, D., Krishnamurthy, A., and Anderson, T.}
\newblock F10: A fault-tolerant engineered network.
\newblock In {\em Proceedings of the 10th USENIX Conference on Networked
  Systems Design and Implementation\/} (Berkeley, CA, USA, 2013), nsdi'13,
  USENIX Association, pp.~399--412.

\bibitem{mordia}
{\sc Porter, G., Strong, R., Farrington, N., Forencich, A., Chen-Sun, P.,
  Rosing, T., Fainman, Y., Papen, G., and Vahdat, A.}
\newblock Integrating microsecond circuit switching into the data center.
\newblock {\em SIGCOMM Comput. Commun. Rev. 43}, 4 (Aug. 2013), 447--458.

\bibitem{htsim}
{\sc Raiciu, C., Wischik, D., and Handley, M.}
\newblock Practical congestion control for multipath transport protocols.

\bibitem{roy2015}
{\sc Roy, A., Zeng, H., Bagga, J., Porter, G., and Snoeren, A.~C.}
\newblock Inside the social network's (datacenter) network.
\newblock In {\em Proceedings of the 2015 ACM Conference on Special Interest
  Group on Data Communication\/} (New York, NY, USA, 2015), SIGCOMM '15, ACM,
  pp.~123--137.

\bibitem{smallWorldDatacenters}
{\sc Shin, J.-Y., Wong, B., and Sirer, E.~G.}
\newblock Small-world datacenters.
\newblock In {\em Proceedings of the 2Nd ACM Symposium on Cloud Computing\/}
  (New York, NY, USA, 2011), SOCC '11, ACM, pp.~2:1--2:13.

\bibitem{JupiterRising15}
{\sc Singh, A., Ong, J., Agarwal, A., Anderson, G., Armistead, A., Bannon, R.,
  Boving, S., Desai, G., Felderman, B., Germano, P., Kanagala, A., Provost, J.,
  Simmons, J., Tanda, E., Wanderer, J., H\"{o}lzle, U., Stuart, S., and Vahdat,
  A.}
\newblock Jupiter rising: A decade of clos topologies and centralized control
  in google's datacenter network.
\newblock In {\em Proceedings of the 2015 ACM Conference on Special Interest
  Group on Data Communication\/} (New York, NY, USA, 2015), SIGCOMM '15, ACM,
  pp.~183--197.

\bibitem{singla14}
{\sc Singla, A., Godfrey, P.~B., and Kolla, A.}
\newblock High throughput data center topology design.
\newblock In {\em 11th {USENIX} Symposium on Networked Systems Design and
  Implementation ({NSDI} 14)\/} (Seattle, WA, 2014), {USENIX} Association,
  pp.~29--41.

\bibitem{singla12}
{\sc Singla, A., Hong, C.-Y., Popa, L., and Godfrey, P.~B.}
\newblock Jellyfish: Networking data centers randomly.
\newblock In {\em Proceedings of the 9th USENIX Conference on Networked Systems
  Design and Implementation\/} (Berkeley, CA, USA, 2012), NSDI'12, USENIX
  Association, pp.~17--17.

\bibitem{OSA}
{\sc Singla, A., Singh, A., and Chen, Y.}
\newblock {OSA}: An optical switching architecture for data center networks
  with unprecedented flexibility.
\newblock In {\em Presented as part of the 9th {USENIX} Symposium on Networked
  Systems Design and Implementation ({NSDI} 12)\/} (San Jose, CA, 2012),
  {USENIX}, pp.~239--252.

\bibitem{flowletSwitching1}
{\sc Sinha, S., Kandula, S., and Katabi, D.}
\newblock {Harnessing TCPs Burstiness using Flowlet Switching}.
\newblock In {\em 3rd ACM SIGCOMM Workshop on Hot Topics in Networks
  (HotNets)\/} (San Diego, CA, November 2004).

\bibitem{openflow_v1.4}
{\sc Specification-Version, O. O.~S.}
\newblock 1.4. 0, 2013.

\bibitem{tomic13}
{\sc Tomic, R.~V.}
\newblock Optimal networks from error correcting codes.
\newblock In {\em Proceedings of the Ninth ACM/IEEE Symposium on Architectures
  for Networking and Communications Systems\/} (Piscataway, NJ, USA, 2013),
  ANCS '13, IEEE Press, pp.~169--180.

\bibitem{valadarsky16}
{\sc Valadarsky, A., Shahaf, G., Dinitz, M., and Schapira, M.}
\newblock Xpander: Towards optimal-performance datacenters.
\newblock In {\em Proceedings of the 12th International on Conference on
  Emerging Networking EXperiments and Technologies\/} (New York, NY, USA,
  2016), CoNEXT '16, ACM, pp.~205--219.

\bibitem{Walraed-Sullivan2013}
{\sc Walraed-Sullivan, M., Vahdat, A., and Marzullo, K.}
\newblock Aspen trees: Balancing data center fault tolerance, scalability and
  cost.
\newblock In {\em Proceedings of the Ninth ACM Conference on Emerging
  Networking Experiments and Technologies\/} (New York, NY, USA, 2013), CoNEXT
  '13, ACM, pp.~85--96.

\bibitem{cthrough}
{\sc Wang, G., Andersen, D.~G., Kaminsky, M., Papagiannaki, K., Ng, T.~E.,
  Kozuch, M., and Ryan, M.}
\newblock c-through: part-time optics in data centers.
\newblock {\em SIGCOMM Comput. Commun. Rev. 41}, 4 (Aug. 2010), --.

\bibitem{MPTCP}
{\sc Wischik, D., Raiciu, C., Greenhalgh, A., and Handley, M.}
\newblock Design, implementation and evaluation of congestion control for
  multipath tcp.
\newblock In {\em Proceedings of the 8th USENIX Conference on Networked Systems
  Design and Implementation\/} (Berkeley, CA, USA, 2011), NSDI'11, USENIX
  Association, pp.~99--112.

\bibitem{Flat-tree}
{\sc Xia, Y., Sun, X.~S., Dzinamarira, S., Wu, D., Huang, X.~S., and Ng, T.
  S.~E.}
\newblock A tale of two topologies: Exploring convertible data center network
  architectures with flat-tree.
\newblock In {\em Proceedings of the Conference of the ACM Special Interest
  Group on Data Communication\/} (New York, NY, USA, 2017), SIGCOMM '17, ACM,
  pp.~295--308.

\bibitem{kShortestPaths}
{\sc Yen, J.~Y.}
\newblock Finding the k shortest loopless paths in a network.
\newblock {\em Management Science 17}, 11 (1971), 712--716.

\bibitem{zhou12}
{\sc Zhou, X., Zhang, Z., Zhu, Y., Li, Y., Kumar, S., Vahdat, A., Zhao, B.~Y.,
  and Zheng, H.}
\newblock Mirror mirror on the ceiling: Flexible wireless links for data
  centers.
\newblock In {\em Proceedings of the ACM SIGCOMM 2012 Conference on
  Applications, Technologies, Architectures, and Protocols for Computer
  Communication\/} (New York, NY, USA, 2012), SIGCOMM '12, ACM, pp.~443--454.

\end{thebibliography}



\appendix

\section{On Long distance links in Expander networks}\label{cabling and datacenter expansion}
The complexity of wiring expander based networks like Jellyfish and Xpander poses a hindrance for adoption. Often a datacenter will be divided into several clusters (for considerations of  incremental expansion, simplifying monitoring and power supply). The metric we consider are number of long distance links or cross cluster links. 

Intuitively, we expect the quality of expander networks to deteriorate with fewer long distance links. We formalize this intuition using the notion of edge expansion $h(G)$ of a graph $G$.

\[h(G) = \text{min}_{|S| < \frac{n}{2}} \frac{\partial S}{|S|}\]

where, $\partial S$ is the number of edges going out of set $S$. Edge expansion can be interpreted as a metric to judge the quality of the expander. Random graphs have high edge expansion ratio~\cite{ellis2011expansion} that asymptotically approaches $d/2$. Theorem~\ref{edge_expansion_limit_theorem} suggests that the edge expansion decreases linearly with the fraction of long distance links.

\begin{theorem} \label{edge_expansion_limit_theorem}
For a $d-$regular graph $G$ with $n$ nodes, partitioned into $k$ equal-sized clusters, \[h(G) \leq \frac{d k f}{2(k-1)}\], where $f=\text{fraction of long distance links}$
\end{theorem}

\begin{proof}
We simply pick $S$ to be all nodes in $k/2$ clusters. There are ${k\choose k/2}$ such sets. Any edge appears in $\partial S$ for $2{k-2\choose k/2-1}$ such sets. Thus, for at least one such set $S$, we will have \[\partial S \leq \frac{ndf}{2} \times 2{k-2\choose k/2-1}\Big/{k\choose k/2} =  \frac{nd k f }{4(k-1)}\].

Thus, \[h(G) \leq \frac{\partial S}{|S|} \leq \frac{nd k f }{4(k-1)}\Big/n/2 = \frac{d k f }{2(k-1)}\]
\end{proof}

\end{document}